\documentclass[11pt]{article}
\usepackage{amssymb,fullpage,soul,nicefrac}
\usepackage{gensymb}
\usepackage{amsfonts}
\usepackage{amsmath}
\usepackage{fullpage}
\usepackage[round]{natbib}
\usepackage{soul}
\usepackage{amsfonts, amsthm, thmtools,graphicx,latexsym,url,epsf, verbatim, setspace, float, paralist, mathtools, authblk, multirow}

\usepackage{thm-restate}

\usepackage{hyperref}

\usepackage{cleveref}

\declaretheorem[name=Theorem,numberwithin=section]{thm}
\declaretheorem[name=Theorem,numberlike=thm]{theorem}
\declaretheorem[name=Lemma,numberlike=thm]{lemma}

\declaretheorem[name=Corollary,numberlike=thm]{cor}

\declaretheorem{corollary}
\declaretheorem{observation}

\declaretheorem{fact}
\declaretheorem{definition}

\newtheorem*{theorem*}{Theorem}

\DeclarePairedDelimiter\ceil{\lceil}{\rceil}
\DeclarePairedDelimiter\floor{\lfloor}{\rfloor}

\usepackage[ruled]{algorithm2e} 

\newcommand{\optoneg}{\textbf{opt}_1}
\newcommand{\opttwog}{\textbf{opt}_2}
\newcommand{\fopt}{\textbf{fopt}}
\newcommand{\opt}{\textbf{opt}}
\newcommand{\lpopt}{\textbf{lpopt}}

\newcommand{\optg}{\textbf{optg}}
\newcommand{\lpoptg}{\textbf{lpopt}}

\newcommand{\Sk}{S_{k}}
\newcommand{\cust}{m}
\newcommand{\supp}{n}
\newcommand{\rahulcom}[1]{\textcolor{blue}{\bf [#1]}}

\newcommand{\qj}{q_{j}}
\newcommand{\xij}{x_{i,j}}
\newcommand{\xijp}{x_{i,j'}}
\newcommand{\Xij}{X_{i,j}}
\newcommand{\Xijp}{X_{i,j'}}
\newcommand{\Mj}{M_{j}}
\newcommand{\M}{M}
\newcommand{\Xj}{X_{j}}
\newcommand{\Xja}{\textbf{X}_{j}^A}
\newcommand{\fnf}{f}
\newcommand{\fng}{g}
\newcommand{\vjp}{v_{j'}}
\newcommand{\qr}{q_{k_2}}
\newcommand{\Xik}{X_{i,k}}
\newcommand{\xj}{x_{j}}
\newcommand{\dconst}{\frac{e-1}{2e}}

\newcommand{\xik}{x_{i,k}}
\newcommand{\Mi}[1]{\mathrm{M}_{i}(#1)}
\newcommand{\yik}{y_{i,k}}
\newcommand{\yi}{y_{i}}

\newcommand{\wk}{w_{k}}
\newcommand{\wl}{w_{k}}
\newcommand{\xikp}{x_{i,k'}}
\newcommand{\wkp}{w_{k'}}
\newcommand{\wlp}{w_{k'}}
\newcommand{\Xk}{\textbf{X}_{k}}
\newcommand{\Yk}{\textbf{Y}_{k}}
\newcommand{\Yja}{\textbf{Y}_{j}^A}
\newcommand{\Ykj}[1]{\textbf{Y}_{k,#1}}
\newcommand{\expt}[1]{\mathbb{E}\left[#1\right]}
\newcommand{\prob}[1]{\mathbb{P}\left(#1\right)}
\newcommand{\xstar}{x^{*}}
\newcommand{\expo}[1]{\exp \left(#1 \right)}
\newcommand{\ckj}{c_{k,j}}

\newcommand{\ckmin}{c_{k,\min}}

\newcommand{\kiran}[1]{\textcolor{blue}{[kiran:#1]}}
\newcommand{\dani}[1]{\textcolor{red}{[Daniela: #1]}}

\newcommand{\vj}{v_{j}}

\newcommand{\defeq}{\stackrel{\mathrm{def}}{=}}
\newcommand{\R}{\mathbb{R}}
\newcommand{\Z}{\mathbb{Z}}
\newcommand{\otcust}{\mathcal{M}}
\newcommand{\otsupp}{\mathcal{N}}
\newcommand{\otbuck}{[1,\buck]_{\Z}}
\newcommand{\Yj}{Y_{j}}
\newcommand{\menui}{\mathrm{M}_{i}}
\newcommand{\menuj}{\mathrm{M}_{j}}
\newcommand{\Mset}{\bar{\mathrm{M}}}
\newcommand{\menuset}{\{\mathrm{M}_{i}\}_{i \in \otcust}}

\newcommand{\ko}{k_{1}}
\newcommand{\kt}{k_{2}}
\newcommand{\Xkj}{\textbf{X}_{k,j}}
\renewcommand{\otbuck}{\textbf{B}}
\newcommand{\ml}{\tm_{\ell}}
\newcommand{\tm}{\textbf{W}}
\newcommand{\Bo}{\textbf{B}_{1}}
\newcommand{\massi}{\alpha_{i}}
\renewcommand{\wl}{w_{k_1}}
\renewcommand{\wlp}{w_{k'_1}}

\newcommand{\Bt}{\textbf{B}_2}
\newcommand{\Mstari}{\mathrm{M}^{*}(i)}

\newcommand{\sk}{s_{k}}

\newcommand{\yil}{y_{i,\ell}}
\newcommand{\opti}{\textbf{opt}(i)}
\newcommand{\copti}{\textbf{copt}(i)}

\begin{document}
\title{Assortment planning for two-sided sequential matching markets}
\date{}
\author[1]{Itai Ashlagi\thanks{iashlagi@stanford.edu}}
\author[2]{Anilesh K. Krishnaswamy \thanks{anilesh@cs.duke.edu}}
\author[3]{Rahul Makhijani \thanks{rahulmj19@fb.com}}
\author[1]{Daniela Saban \thanks{dsaban@stanford.edu}}
\author[1]{Kirankumar Shiragur \thanks{shiragur@stanford.edu}}

\affil[1]{Stanford University}
\affil[2]{Duke University}
\affil[3]{Facebook}

\setcounter{Maxaffil}{0}
\renewcommand\Affilfont{\small}
\maketitle
Two-sided matching platforms provide users with menus of match recommendations. To maximize the number of realized matches between the two sides (referred herein as customers and suppliers), the platform must balance the inherent tension between recommending more suppliers to customers for potential matches, and avoiding collisions that arise when customers are given more choice (and end up choosing the same suppliers). We introduce a stylized  model to study the above trade-off. The platform offers each customer a menu of suppliers, and customers choose, simultaneously and independently, to either select a supplier from their menu or remain unmatched. Suppliers then see the set of customers that have selected them, and choose to either match with one of these customers or remain unmatched. A match occurs if a customer and a supplier choose each other (in sequence). 
Agents' choices are probabilistic, and proportional to the public scores of agents in their menu and a score that is associated with the outside option of remaining unmatched.
The platform's problem is to construct menus for  customers to maximize the total number of matches. 
We first show that this problem is strongly NP-hard, and then provide an efficient algorithm that achieves a constant-factor approximation to the optimal expected number of matches. Our algorithm uses bucketing techniques (grouping similar suppliers into buckets), together with a linear programming based relaxation and rounding. We finally provide simulations to better understand how the algorithm might behave in practice.

\section{Introduction}
\label{sec:introduction}

Two-sided platforms, that enable agents from one side to match and transact with agents on the other side, 
are now ubiquitous in a variety of markets such as those for labor, dating, accommodation rentals, etc. Many such online platforms, rather than determining who matches with whom in a centralized fashion, operate by presenting users with a set of recommended partners and allowing them to choose which of these to pursue a match with. For example, when looking for accommodation on Airbnb, a potential guest is immediately shown a map containing up to twenty properties available in the time frame she selects; the guest can then choose which of these property owners to contact. Similarly, clients looking for freelancers on Upwork or dog owners looking for a sitter on Rover are also presented with a set of options for them to consider, before deciding which one to reach out to. 
Therefore, a question that is central to the operations of many of these platforms is to optimize the recommendations over potential partners that are presented to its users.

A natural objective in many of the aforementioned markets is to maximize the total number of matches (e.g., number of bookings, tasks completed) that occur through the platform. A key feature that the platform must account for is that a match between two users  occurs only when both users find each other acceptable. For instance, once a dog sitter is contacted, she needs to be willing to perform the job; similarly, a host offering a vacation rental needs to accept the guest's request to book the same. Moreover, when many conflicting requests are received at the same time, at most one such request can be accepted. 

To describe the trade-off that then arises, let us refer to the two sides of the market by customers and suppliers, respectively, where we assume that costumers (potential guests,  clients, dog owners) are presented with recommended assortments of potential suppliers (properties, freelancers, dog sitters) to choose from. Naturally, increasing the number and mix of potential suppliers presented to a costumer increases the chances that she finds an acceptable partner and  reaches out to one of them instead of choosing not to pursue a match through the platform at all. However, if as a result the same supplier is shown to many customers, this  also increases the chances that a supplier is contacted by many customers with conflicting requests at the same time, resulting in \emph{collisions} as only one  such request can be accepted. In a setting where suppliers are indeed products (i.e., in commodity markets) and do not have preferences, these collisions tend to bad for the platform, as the platform could naturally increase the number of transactions by redirecting some of these costumers to other products. However, in a two-sided market, having a supplier receiving many requests might indeed increase the number of matches as this may increase the chances she finds a customer she prefers over her outside option.

Motivated by the above discussion, this paper studies how to optimize over the assortments of recommended potential partners that are shown to each user so as to maximize the number of matches. Our contributions are along two lines: we propose a novel stylized model that captures the aforementioned tradeoffs, and we propose a simple algorithm  that  achieves a constant-factor approximation to the optimal number of matches.

In our a stylized model, suppliers are heterogeneous across two dimensions: their public  ``attractiveness" to customers and the value for their outside option, that is known to the platform.  Matching decisions occur \textit{sequentially} in two stages as follows. First, each costumer is presented with a menu of suppliers chosen by the platform. Costumers choose, independently, either one supplier from the menu or her outside option (i.e., to remain unmatched); following the tradition of the assortment planning literature, we assume that choice model is given by the multinomial logit (MNL). In the second stage, each supplier is presented with a menu containing  only those customers who have selected her in the first place. Each supplier then chooses either to match with one of these customers or to use her outside option.  
A customer and a supplier  match only if they both chose each other. The platform's \textit{ two-sided sequential assortment problem} is to choose a menu of assortments to show to each customer so as to maximize the expected number of matches. 

Our first result is to establish that the platform's optimization problem is NP-hard. Following this hardness result, we propose a polynomial-time algorithm that achieves a constant-factor approximation ratio to the optimal number of expected matches. Our focus is not on finding the best approximation ratio but rather on designing an intuitive and easy to implement algorithm, which somewhat resembles the current implementation of the assortment decisions in many marketplaces, and to establish some performance guarantee on such algorithms.  
To design the algorithm, we focus on two regimes: a \textit{low-value suppliers' regime},  where all suppliers are less attractive than the customers' outside options, and a \textit{high-value suppliers' regime}, where all suppliers are at least as attractive as the customers' outside options.
We provide an efficient constant factor approximation algorithm for constructing the menu sets in each regime. Combining these algorithms (in a black box fashion) suffices to establish a constant-factor approximation for the general case.

The high-value suppliers' regime is easier to analyze. In this regime, a customer is at least as likely to choose any given supplier over her outside option. Informally, this implies that the platform does not need to provide a lot of options to consumers in order for them to prefer one of the suppliers in the platform over their outside option.  In fact, we show that a menu set in which each costumer is presented with a single supplier suffices to achieve a constant-factor approximation guarantee. 

The case with low-value suppliers is likely to be a better representation of the reality in many online markets. In this setting, the outside option is (ex-ante) better than an any individual option the platform can offer; this can be the case if, for example, the platform has a somewhat strong competitor operating in the same space. Therefore, in principle, the platform might need  to offer larger assortments to the consumers, which is consistent with what is observed in practice.  Our algorithm works by \textit{bucketing} (grouping) suppliers that have similar attractiveness scores and similar value of their outside option in the same bucket. Using this buckets, we propose an LP relaxation that  provides an upper bound (up to constant factor) to the expected number of matches. Our LP treats  all  suppliers  in  the  same  bucket  as  indistinguishable  and   tracks how  many  suppliers  in each bucket are shown to each user to maximize the expected number of matches.\footnote{This, for example, resembles what some dating platforms do, where users are assigned a unique ELO-type score, and assortments are decided using that score. Source: \url{https://www.theverge.com/2019/3/15/18267772/tinder-elo-score-desirability-algorithm-how-works}} We then propose a rounding algorithm that rounds the solution to the aforementioned LP and use the output to construct the assortments for the customers. The key idea behind the menu constructing algorithm is to show each supplier in the same bucket to \emph{approximately} the same number of users. This turns out to be useful to minimize collisions, i.e., situations where one supplier is being selected by too many customers while a similar supplier is selected by too few of them. 
Moreover, our algorithm has the intuitive property that the number of menus in which a supplier appears is upper bounded by a constant that depends both on her  value and the  value of her outside option. In particular, suppliers with lower values can be shown more than those with the same outside option and higher values; this is because suppliers with lower values are less likely to be selected by consumers. On the other hand, comparing suppliers with the same value, those with higher outside options can be shown more often; this implies that, in expectation, such a supplier will be selected by more consumers, increasing the chances that she will choose one of them to match in the platform. 

The remainder paper is organized as follows. The model and problem are formalized in Section \ref{sec:model}. Section \ref{sec:results} states an overview of the main results and techniques. Sections \ref{subsec:smallSupplierLargeOutside} and \ref{subsec:largeSupplier} provide the main steps of proofs for low- and high-value suppliers' regimes, respectively (with some details relegated to the appendix). In Section~\ref{subsec:simulations}, we provide simulation results to illustrate the performance of the algorithm in some randomly generated instances. Finally, in Section \ref{sec:conclusion} we summarize our findings and discuss a range of open questions.

 \subsection{Related literature}
The paper is related to several strands in the literature. The first strand is the literature on assortment optimization in the context of revenue management. Following the seminal paper by \citet{ryzin}, which studies an  assortment optimization problem for inventory selection using the multinomial-logit framework, there has been a vast body of work in this area models (see, e.g., \citet{rusmevichientong2014assortment,mendez2014branch,bront2009column,gaur2006assortment,immorlica2018combinatorial} and a survey by \citet{kok2008assortment}).
The key difference between these studies (which address a variety of choice models, dynamics, constraints, etc.), and our paper is that they consider one-sided markets in which the goal is to assign agents a menu of goods (who have no preferences) in order to maximize revenue.  In general, even one-sided assortment planning problems are computationally hard, except for some notable cases such as MNL without constraints on the offered sets  \citet{talluri2004revenue}. 
Closer to our setting, are those papers than incorporate capacity constraints bounding the number of goods of each type, \citep{rusmevichientong2010dynamic,davis2013assortment}; these capacities constraints are inherent in our model, as users can only match once. 

Preferences on both sides of the market add a new layer in comparison to the  assortment planning problem in commodity markets. While the literature of the latter problem is  mature,  little is known about assortments in two-sided matching markets; see \citet{shi2016assortment,ashlagi2015optimal} for studies  in large two-sided markets.

Another  stream of related papers  considers the relationship between choices  and congestion in two-sided matching markets. The benefit from limiting the choices offered to users in a two-sided market already appears in several studies \citep{halaburda,arnosti2014managing,kanoria2017facilitating}. However, these papers either do no explicitly allow the platform to select an assortment for the users, or the users are ex-ante homogeneous in the eyes of the platform.  
Our paper complements this line of work by considering a stylized static model while allowing for more heterogeneity of agents and using the assortment shown to each user as a platform lever. 

More generally, this paper is related to the two-sided matching mechanisms. Stable mechanisms typically elicits agents' ranking lists and form stable matches. In our problem a model of choice is estimated by the designer  and (only first) choices are done in a decentralized manner. Our choice model, however, is similar to the a preference formation in several papers in this literature \citep{immorlica2005marriage,kojima2009incentives,ashlagi2014stability}.

Finally, we note the key difference between our problem and (online) stochastic bipartite matching \citep{feldman2009online,chen2009approximating, manshadi2012online,haeupler2011online}. In this literature the planner makes matches given some distribution over edges. In our model, the platform only offers a set and  users to \textit{choose} from, which inherently creates  collisions due to lack of coordination.

 \subsection{Related literature}

The paper is related to several strands in the literature. The first strand is the literature on assortment optimization in the context of revenue management. Following the seminal paper by \citet{ryzin}, which studies an  assortment optimization problem for inventory selection using the multinomial-logit framework, there has been a vast body of work in this area models (see, e.g., \citet{rusmevichientong2014assortment,mendez2014branch,bront2009column,gaur2006assortment,immorlica2018combinatorial} and a survey by \citet{kok2008assortment}).
The key difference between these studies (which address a variety of choice models, dynamics, constraints, etc.), and our paper is that they consider one-sided markets in which the goal is to assign agents a menu of goods (who have no preferences) in order to maximize revenue.  In general, even one-sided assortment planning problems are computationally hard, except for some notable cases such as MNL without constraints on the offered sets  \citet{talluri2004revenue}. 
Closer to our setting, are those papers than incorporate capacity constraints bounding the number of goods of each type, \citep{rusmevichientong2010dynamic,davis2013assortment}; these capacities constraints are inherent in our model, as users can only match once. 

Preferences on both sides of the market add a new layer in comparison to the  assortment planning problem in commodity markets. While the literature of the latter problem is  mature,  little is known about assortments in two-sided matching markets; see \citet{shi2016assortment,ashlagi2015optimal} for studies  in large two-sided markets.

Another  stream of related papers  considers the relationship between choices  and congestion in two-sided matching markets. The benefit from limiting the choices offered to users in a two-sided market already appears in several studies \citep{halaburda,arnosti2014managing,kanoria2017facilitating}. However, these papers either do no explicitly allow the platform to select an assortment for the users, or the users are ex-ante homogeneous in the eyes of the platform.  
Our paper complements this line of work by considering a stylized static model while allowing for more heterogeneity of agents and using the assortment shown to each user as a platform lever. 

More generally, this paper is related to the two-sided matching mechanisms. Stable mechanisms typically elicits agents' ranking lists and form stable matches. In our problem a model of choice is estimated by the designer  and (only first) choices are done in a decentralized manner. Our choice model, however, is similar to the a preference formation in several papers in this literature \citep{immorlica2005marriage,kojima2009incentives,ashlagi2014stability}.

Finally, we note the key difference between our problem and (online) stochastic bipartite matching \citep{feldman2009online,chen2009approximating, manshadi2012online,haeupler2011online}. In this literature the planner makes matches given some distribution over edges. In our model, the platform only offers a set and  users to \textit{choose} from, which inherently creates  collisions due to lack of coordination.  

\section{Model}
\label{sec:model}



In this section, we introduce a stylized model to capture the problem faced by a platform that must recommend assortments to her customers so as to maximize the number of matches in a two-sided sequential matching market.
 We start by formally introducing the model. Afterwards, we discuss some of the modeling assumptions.

\paragraph{Problem setup.} We consider a platform, consisting of a set of customers and suppliers, who must construct assortments of suppliers to recommend to her costumers. In our simplified setting, events occur sequentially as follows. 
First, the platform chooses the menu profile, which consist of a menu (subset) of suppliers that will be shown to each customer. Second, customers observe their menus and, simultaneously and independently, choose at most one supplier from their menu. Third, each  supplier observes all customers who chose her, and chooses at most one customer from this set to match with. (Below, we describe how customers and suppliers make their choices.) When a customer is chosen by a supplier, this results in a match. The platform chooses a menu profile to maximize the expected number of matches.


In more detail, there is a set of customers indexed by $\otcust=\{1,\ldots,\cust\}$, and a set of suppliers, $\otsupp=\{1,\ldots,\supp\}$, affiliated to the platform. The platform chooses a \textit{profile of menus} $\Mset=\menuset$, where $\menui \subseteq \otsupp$ is the subset of suppliers that is shown to customer $i\in \otcust$. We refer to the menu profile also as a {\it menu set}.

We assume each supplier  $j\in\otsupp$ has an associated  score $v_j>0$, which represents the (ex-ante) consumer preference values. 
Moreover,  every customer has an outside option indexed by $0$; we assume all customers' outside options score are identical and normalized to $v_0=1$.   Given a menu profile $\menuset$, each customer $i \in \otcust$  \emph{simultaneously} and independently chooses at most one supplier from the her menu $\menui$ according to a
multinomial logit model. That is, given that customer $i$ is assigned menu  $\menui$, she chooses  $j\in \menui\cup \{0\}$ with probability 

\begin{equation}
    \label{eqn:RUM} p_{ij}(\menui)  = \frac{v_j}{\sum_{k \in \menui}v_k +v_0}. 
\end{equation}
 
Hence, the probability that  $i$ selects a supplier from menu $\menui$  is given by
\begin{equation}
    \label{eq:rum_expected number of matches}
    \frac{\sum_{k \in \menui} v_k}{\sum_{k \in \menui} v_k +v_0}.
\end{equation}

Once customers' choices are made, each supplier $j\in\otsupp$ is presented with a menu $\menuj =\{i \in \otcust:~ j \in \menui \textup{ and }i \textup{ chose supplier }j\}$. Given these menus, each supplier $j\in \otsupp$ chooses $i\in\menuj \cup \{0\}$ where $0$ denotes an the supplier's outside option. As consumers, we also assume that suppliers also make choices using a multinomial logit model.   However,  we make two different assumptions. First, every customer $i\in \otcust$ has the same public score normalized to $q_i=1$. Second, we allow suppliers to have different outside options. That is, every supplier  $j\in \otsupp$ has an outside option score $q_j\geq 0$.  Then, if supplier $j$ observes menu $\menuj\neq \emptyset$, where $\menuj$ is the set of customers that chose $j$, the probability that  $j$ chooses any given customer $i \in \menuj$  is  $\frac{1}{|\menuj| +q_j}$, and the probability that  $j$ chooses any customer (as opposed to choosing her outside option) is
    $\frac{|\menuj|}{|\menuj|+q_j}$.  
    Note that if  no customer chose supplier $j$,  $j$ must remain unmatched.
    
We assume that the choice models of suppliers and customers (public and outside option scores) are  known to the platform. The platform must select a  menu profile to maximize the expected number of matches.  Formally, we define the  \textit{two-sided sequential assortment problem} as follows:

\begin{definition}[Two-sided sequential assortment problem]
Given a set of customers $\otcust$ and set of suppliers $\otsupp$ with associated non-negative public scores  $\{v_j\}_{j\in \otsupp}$ and non-negative outside option scores $\{q_j\}_{j\in \otsupp}$, the platform chooses a menu profile $\menuset$, which consists of a menu of suppliers to be shown to each customer, in order to maximize the expected number of matches.
\end{definition}

\paragraph{Modeling tradeoffs and assumptions.}
The above model, while simple, captures the main dynamics and tradeoffs faced by a platform when recommending suppliers to her consumers to maximize the number of  resulting matches. 
The first step towards match formation is to ensure that consumers select a supplier in the platform. Recommending better and more suppliers to each individual consumer increases the chances of each of them finding a supplier they like over their outside option. However, as a supplier gets  recommended to more consumers, this increases the chances that she gets selected by multiple such consumers, and these chances increases with the quality of the supplier. This has two effects: it increases the probability of the supplier matching, as it increases the changes the supplier chooses one of these consumers instead of her outside option, but it may decrease the total number of matches as each supplier can match at most once (even if she was selected by multiple consumers). 

\noindent \textit{Simultaneous choices.}
Note that we assume that customers make their choices simultaneously. If we consider instead a dynamic problem where each customer arrives, chooses a supplier and the supplier then immediately responds, then this problem is similar to previous work on dynamic assortment planning with inventory constraints (see, e.g., \cite{golrezaei2014real}). However, simultaneous choices by all agents on the same side allow us to represent a settings  where there might be a delay between the costumer choice and the supplier response, for example, the supplier might check the app once a day or, as in the case of some dating platforms, users see assortments only once per day.

\noindent \textit{User heterogeneity.}
Note that we assumed that suppliers are ex-ante heterogeneous across two dimensions: their scores $v_j$s and the values of their outside option, $q_j$s. By contrast, consumers are (ex-ante) identical from the platform's perspective: we assume that each has a value $v_i=1$ and the value of their outside option is $v_0=1$. This is a reasonable assumption in many settings, where the platform may have detailed information about the suppliers while she might either have very little to no information about consumers, or about suppliers' preferences. For example, in a setting like Airbnb, the platform and the consumers can observe many metrics of the suppliers' qualities (property  photos, location, and even past reviews) whereas the information regarding consumers might not be as detailed, and even such information may be less relevant for the suppliers' choice. However, we believe that using similar bucketing ideas, a constant-factor approximation guarantee can be achieved in the case of consumers having different scores as well.

\section{Overview of results}
\label{sec:results}

In this section, we present an overview of the main results. Throughout the paper, we use boldface to refer to random variables. We also use boldface to refer to name the optimization problems presented later. In Appendix~\ref{app:tools}, we state some standard results that will be used in our proofs.

Our first result is to show that the two-sided sequential assortment problem is NP-Hard.  
\begin{restatable}{proposition}{nphardness}
\label{prop:sequential_np_hard}
 The two-sided sequential assortment problem is strongly NP-hard.
\end{restatable}

The proof\footnote{The authors would like to thank Vineet Goyal for useful discussions and valuable suggestions.} can be found in Appendix \ref{appendix_results}. There, we use a reduction from the $3$-partition problem, which is known to be strongly NP-complete \citep{garey2002computers}. This result is not surprising, given that assortment planning problems tend to be computationally hard even in simpler one-sided settings.

The hardness result above motivates us to find an alternative solution. We propose simple and intuitive algorithms, and then show that these indeed achieve constant-factor approximation guarantees. 

Formally, we say that a menu construction  algorithm $\mathbb{A}$ has an \textit{approximation ratio} $\alpha$ if the expected number of matches ($\textbf{M}_{\mathbb{A}}$)  achieved by the menu set returned by algorithm $\mathbb{A}$ is at least $\alpha$ times the optimal expected number of matches $\opt$. Formally,
 $$ \expt{\textbf{M}_{\mathbb{A}}} \geq \alpha \cdot \opt .$$
In our setting, we say that an algorithm is \emph{efficient} if it runs in polynomial time in the number of suppliers and customers. 





 
 Our main result is to provide an efficient constant-factor approximation algorithm, as summarized by the following theorem.
 
 \begin{restatable}[General case]
 {thm}{thmmain}
\label{thm:main}
There exists an efficient algorithm for the sequential assortment problem that runs in polynomial time (in the number of suppliers and customers) and that achieves a constant approximation ratio. 
\end{restatable}

In order to prove the above result, we analyze the problem in two different regimes, each of which is interesting in its own right. The first regime, the \textit{low-value supplier regime},  is when all suppliers have scores $\vj< 1$;  the other regime, the \textit{high-value supplier regime}, is when all suppliers have scores $\vj \geq 1$. For each of these regimes, we provide an efficient constant factor approximation algorithm for constructing  the menu sets. Later, we propose an algorithm that combines the algorithms from both these regimes (in a black box fashion), which achieves a constant factor approximation to the general case.

The algorithms for both regimes follow a similar high-level strategy. First, we provide an optimization problem that can be solved efficiently and is also an upper bound (up to a constant factor) to our original objective.  We then construct menu sets based on the solutions to these (easy to solve) optimization problems, that further achieve a constant factor approximation to the upper bounds and therefore a constant factor approximation to the original problem. Although, the algorithms for the two special cases follow similar a strategy, the optimization problem that upper bounds the original objective and the menu set construction algorithms are entirely different. We describe these in some detail next. 

\paragraph{Low-value suppliers' case.}
The case with low-value suppliers ($\vj <1$) is likely to be a better representation of the reality in many online platforms, but it is also a more technically challenging case. Recall that we have normalized the customers outside option to be $v_0=1$. Therefore, the case with $\vj <1$ represents a setting where the outside option is (ex-ante) better than an any individual option the platform can offer, and thus the platform should, in principle, offer larger assortments to consumers. One possible scenario that this can represent when the platform has a strong competitor, which is captured in reduced form by the customers' outside option. Thus, after a search, consumers are presented with many options. (For instance, after a search in Airbnb, the default is for the potential guest to be presented with 20 options per page.) 
Next we state the result that gives an efficient constant factor approximation algorithm for this case. 
\begin{restatable}[Low-value suppliers]{thm}{thmcaseone}
\label{thm:caseA}
Suppose that $v_j \leq 1 $ for all $j \in \otsupp$. Then, there exists an efficient algorithm for constructing menus for each customer $i \in \otcust$ such that the expected number of matches  is at least a constant fraction $\alpha_L$ of the expected number of matches achieved by the optimal menu set, where $\alpha_L$ is some absolute constant.
\end{restatable}

To prove the above theorem, we first construct an LP optimization problem (Section \ref{sec_lp_relax_gen}), one whose optimal solution upper bounds the expected number of matches obtained by an optimal menu set (Lemmas \ref{lem:foptg} and \ref{lem:lpoptg}). Crucial to defining this LP is a two-dimensional \emph{bucketing} (grouping) of the suppliers. In particular, we group suppliers that have similar scores $v_j$ and similar value of their outside option $q_j$ in the same bucket. When constructing  the LP, we treat  all  suppliers  in  the  same  bucket  as  indistinguishable  and  we  track  how  many  suppliers  in each bucket are shown to each user.
We propose a rounding algorithm (Algorithm \ref{alg:gone}) that returns an integer solution and approximates the optimal solution of the above LP (Lemma \ref{lem:gone}) by a constant factor. We then use the output solution from the rounding algorithm as an input to Algorithm \ref{alg:gmenu}, which constructs the menus for the customers. In other words, the algorithm takes as an input how many suppliers in each bucket to be shown to each user,  and construct the actual menus by choosing \emph{which} suppliers are shown to each user. The key idea behind the algorithm is to show each supplier in the same bucket to \emph{approximately} the same number of users, which is useful to minimize the probability that cases where one supplier ends up being selected by too many customers while the other supplier is selected by too few of them arise. Moreover, our algorithm has the intuitive property that, for every supplier, the number of menus in which she appears is upper bounded by a constant that depends both on her (representative) value and the (representative) value of her outside option.

Finally, to establish Theorem \ref{thm:caseA}, we show how the menus thus constructed using Algorithm \ref{alg:gmenu} result in an expected number of matches that are within a certain constant factor $\alpha_L$ of the LP relaxation upper bound and thus of the optimal menu set. We discuss the above techniques more formally in Section \ref{subsec:smallSupplierLargeOutside}. 

\paragraph{High-value suppliers' case.}
In the high-value supplier regime, all suppliers $j\in \otsupp$ have scores $\vj \geq 1$. As before, we also show that we can construct an efficient constant-factor approximation algorithm for this case.
\begin{restatable}[High-value suppliers]{thm}{thmcasetwo}
\label{thm:caseB}
Suppose that $v_j \geq 1$ for all $j \in \otsupp$. Then, there exists an efficient algorithm for constructing menus for each customer $i \in \otcust$, such that, the expected number of matches is at least a constant fraction $\alpha_H$ of the expected number of matches achieved by the optimal menu set, where $\alpha_H$ is some absolute constant.
\end{restatable}

For this case, we employ a different approach to that in  the proof of Theorem \ref{thm:caseA}. Recall that we have normalized the value of customers' outside option to be one. 
Therefore, in the high-value suppliers regime, a customer is at least as likely to choose any given supplier over her outside option. Informally, this implies that the platform does not need to provide a lot of variety to consumers in order for them to prefer one of the suppliers in the platform over their outside option.  
In order to establish the result, we first construct a combinatorial optimization problem that upper bounds the expected number of matches achieved by the optimal menu set (Lemma \ref{lem:casetwoup}). This combinatorial optimization problem can be solved approximately (up to factor of $1/2$) in polynomial time (Lemma~\ref{lem:eff}). Finally, we construct menu set based on the approximate solution, which can then be guaranteed to achieve a constant factor approximation of the optimal expected number of matches possible in this setting (Theorem \ref{thm:caseB}). The menu set constructed by the algorithm results in each costumer being presented with a single supplier, which basically helps the platform to avoid lack of coordination on the consumer side. In our algorithms and results,  we extensively exploiting the fact that suppliers are of high value. For instance, this allows us to establish that the proposed upper bound is be relatively tight. A detailed discussion of the proof of the above theorem can be found in Section \ref{subsec:largeSupplier}. 

\BlankLine

Finally, we show that the results for the two different cases above can be combined to give a constant factor approximation to the original problem as showcased by the following result, the proof of which is deferred to Appendix \ref{appendix_results}.

\begin{restatable}{thm}{thmcombine}
\label{thm:combine} Fix the set of customers $\otcust$ and set of suppliers $\otsupp$ with associated non-negative public scores  $\{v_j\}_{j\in \otsupp}$ and non-negative outside option scores $\{q_j\}_{j\in \otsupp}$.  
Suppose that there exist efficient algorithms that achieve an approximation ratio of $\alpha_L$ and $\alpha_H$  for the low- and high-value supplier cases, respectively. Then, there exists an efficient algorithm for the general two-sided sequential assortment optimization problem such that the expected number of matches is at least a fraction $\frac{1}{2}\min\{\alpha_L,\alpha_H\}$ of the expected number of matches achieved by the optimal menu set.
\end{restatable}


To conclude, note that \Cref{thm:main} follows immediately by the  previous results as, by Theorem~\ref{thm:combine}, the existence of the efficient constant-factor approximation algorithms for the low- and high-value supplier cases (Theorems \ref{thm:caseA} and \ref{thm:caseB}) guarantees the existence of a constant-factor approximation algorithm for the general case.

\section{Proof overview for the low-value suppliers' case}\label{subsec:smallSupplierLargeOutside}

In this section, we analyze the case where all suppliers $j \in N$ have scores $\vj \leq 1$.
We provide an algorithm that outputs menu set $\menuset$, such that the expected number of matches under $\menuset$ is a constant-factor approximation to the expected number of matches achieved by the optimal menu set. 

The organization of this section is follows.  In \Cref{sec:preliminaries} we introduce notations and a preliminary result. In \Cref{sec_lp_relax_gen} we provide a linear program using the idea of bucketing that upper bounds the expected number of matches  achieved by the optimal menu. This LP outputs a fractional solution; in \Cref{rounding_general} we use this fractional solution as an input to a rounding algorithm that outputs an integral solution and use this solution to construct a menu set that achieves a constant-factor approximation to the expected number of matches  achieved by the optimal menu.

\subsection{Preliminaries}\label{sec:preliminaries}
We now provide the notations useful for the purpose of our analysis. Let $\opt$ be the expected number of matches achieved by the optimal menu set.
In the rest of the paper, we slightly abuse notation and, when clear from the context, we refer to an optimization problem and to its optimal value using the same name. 

To simplify the exposition in the remaining of the section, we assume that 
 all the suppliers have outside option scores greater than one, i.e.,  for all $j \in \otsupp$ we have  $\qj\geq 1$. As we show by the end of this subsection, this is without loss of generality up to constant factors. 

To establish our result, we use the idea of \emph{bucketing}, where we group suppliers into buckets, where each bucket includes all suppliers that are ``similar" in terms of both their value (public score) and the value of their outside option. 
Therefore, we use a two-dimensional index $k=(k_1,k_2) \in \Z_{\geq 0}^{2}$ for the buckets. The $k$-{th} bucket is given by $$\Sk \defeq \left\lbrace j \in \otsupp:~\vj \in \left[\frac{1}{2^{k_1}},\frac{1}{2^{k_1-1}}\right)_{\R} \text{ and } \qj \in [2^{k_2},{2^{k_2+1}})_{\R}\right\rbrace.$$ 
For bucket $k$, we define a \textit{representative score} $\wl$ and a \textit{representative value of the outside option}  as
\begin{equation}\label{eq:rep_scores}
    \wl \defeq \frac{1}{2^{\ko}} \qquad \mbox{ and } \qquad \qr \defeq 2^{\kt},
\end{equation}
respectively. 
 Note that for any $j\in \Sk$, we have that $\wl \leq \vj \leq 2\wl$ and $\qr \leq \qj \leq 2\qr$. We use $\Bo$ and $\Bt$ to denote the set of all $\wl$ and $\qr$ values, respectively, corresponding to non-empty buckets. Further, let the set of all non-empty bucket indices be denoted by $\otbuck$, that is,  $\otbuck\defeq \{k=(\ko,\kt)\in \Bo \times \Bt:~|\Sk|\geq 1\}$.  

To conclude, we show that it is without loss (up to constant factors) to  assume that all suppliers have $q_j \geq 1$. To show this, we consider the following two settings. 

\noindent {\bf Original setting}: Let $\otcust$ and $\otsupp$ be the set of customers and suppliers, respectively. For each supplier $j \in \otsupp$  let $\vj \leq 1$ and $\qj$  be her score and outside option, respectively. 

We define a new setting by modifying the outside option scores for a subset of suppliers. For the new setting let $S \defeq \{j\in \otsupp~|~\qj <1\}$.

\newcommand{\qpj}{q'_{j}}
\newcommand{\vpj}{v'_{j}}
\noindent {\bf New setting}: Consider an instance from the original setting, and define an instance for the new setting where, $\vpj$ and $\qpj$, the score and outside option of supplier $j\in \otsupp$ in the new setting are defined as follows:
\begin{align*}
\vpj\defeq\vj \text{ for all }j\in \otsupp \quad \quad \quad \qpj \defeq 
\begin{cases}
1 \text{ if } j \in S,\\
\qj \text{ otherwise}.
\end{cases}
\end{align*}

We now state the following result, which is proved in \Cref{appendix_general}.

\begin{restatable}{lemma}{lemcaseone}
    \label{lem:case11}
 Fix any feasible menu set  $\menuset$ for customers,  and let $\textbf{M}$ and $\textbf{M}'$ be the random variables counting the number of matches obtained when using  $\menuset$ in the {\bf original} and in the {\bf new} setting, respectively. Then,
$$\expt{\textbf{M}'} \leq \expt{\textbf{M}}\leq 2 \expt{\textbf{M}'}.$$
\end{restatable}

\subsection{An upper bound}
\label{sec_lp_relax_gen}

The objective of this subsection is to  provide a linear program that, up to constant factors, upper bounds the expected number of matches achieved by the optimal menu set. The proposed linear program essentially chooses, for each customer, how many suppliers in a given bucket to show to her. 

In route towards defining our linear program, for each customer $i \in \otcust$ and  each supplier $j \in \otsupp$, define a variable $\xij$ which takes value $1$ if $j$ is present in $i$'s menu and $0$ otherwise. In addition, for  every $i\in \otcust$ and $k \in \otbuck$ define a variable $\xik$ which counts the number of suppliers from bucket $k$ present in $i$'s menu, i.e., $\xik\defeq\sum_{j\in \Sk}\xij$ and  thus $0 \leq \xik \leq |\Sk|$. Using this definition, we will derive some upper bounds on the expected number of matches that will constitute the basis for our linear program.

Consider a the menu set given by $x = \{\xij\}_{i\in \otcust, j \in \otsupp}$. 
Let $\Xk$ for all $k=(\ko,\kt)\in \otbuck$ be the random variable counting the number of customers that choose a supplier from bucket $k$. Similarly, let $\Xkj$ be the random variable counting the number of customers that choose supplier $j$ from bucket $k$. Observe that
    \begin{equation}\label{eq:xkj}
    \Xk=\sum_{j \in \Sk} \Xkj.    
    \end{equation}
By linearity of expectation and the definition of the representative scores $w$s (see \eqref{eq:rep_scores}), the expected value of $\Xk$ satisfies
\begin{equation}\label{eq:gexptxk}
    \expt{\Xk}= \sum_{j \in \Sk} \expt{\Xkj}= \sum_{j\in \Sk}\sum_{i \in \otcust}\frac{\vj \xij}{1+\sum_{j'\in \otsupp} \vjp \xijp} \leq 2\sum_{i \in \otcust}\frac{\wl \xik}{1+\sum_{k' \in \otbuck}\wlp \xikp}~.    
     \end{equation}
where we used that $v_j\leq 2 \wl$ for all $j \in \Sk$ and that $\sum_{j'\in \otsupp} \vjp \xijp \geq \sum_{k' \in \otbuck}\wlp \xikp$. 

For $k\in \otbuck$ and  supplier $j$ in bucket $k$, let $\Ykj{j}$ denote the random variable  representing whether supplier $j$ is matched, and  let $\Yk$ be the random variable counting the number of matches in the platform involving suppliers from bucket $k$. 
The expected value of $\Yk$ satisfies
    \begin{equation}\label{eq:gexptyk}
    \expt{\Yk} = \expt{\sum_{j \in \Sk} \expt{\Ykj{j}|\Xkj}}= \expt{\sum_{j \in \Sk}\frac{\Xkj}{\qj+\Xkj}} \leq \expt{\sum_{j \in \Sk}\frac{\Xkj}{\qr+\Xkj}} ~.    
    \end{equation}
The second equality follows from the suppliers' choice model and the fact that suppliers value all customers who chose them equally at $v_i=1$, and  the inequality follows since $\qj \geq \qr$ for all $j\in \Sk$. 

We next provide an upper bound on the expected value of $\Yk$ that turns out to be easier to work with for optimization purposes. Note that, for each $j \in \Sk$, the expression $\frac{\Xkj}{\qr+\Xkj}$ is upper bounded by $\min \left \{ \frac{\Xkj}{\qr},1 \right \}$. Further, by Jensen's inequality we have  $\expt{\min \left \{ \frac{\Xkj}{\qr},1 \right \}}\leq \min \left \{\expt{ \frac{\Xkj}{\qr}},1 \right \}$. Combining  the  above two observations together with the linearity of expectation and Equation \eqref{eq:gexptyk}, we obtain that
\begin{equation}
\label{eq:exptykub}
    \begin{split}
\expt{\Yk}&\leq \sum_{j \in \Sk}\min \left \{\expt{ \frac{\Xkj}{\qr}},1 \right \} \leq \min \left \{\sum_{j \in \Sk}\expt{ \frac{\Xkj}{\qr}},|\Sk| \right \} \\
&=\min \left \{\frac{1}{\qr}\expt{\Xk},|\Sk| \right \} \leq \min \left \{\frac{2}{\qr}\sum_{i=1}^{m}\frac{\wl \xik}{1+\sum_{k' \in \otbuck}\wlp \xikp},|\Sk| \right \}, 
    \end{split}
\end{equation}

\noindent where the second inequality follows because $\min\{.\}$ is a concave function and the last inequality follows from Equation \eqref{eq:gexptxk}. Note that the variables $\xik$ are sufficient (and we don't need $\xij$ variables) to provide an upper bound on the expected number of matches. Therefore, for our purposes, we will treat all suppliers in the same bucket as indistinguishable and we track how many suppliers in bucket $k$ are shown to each user. 

Based on the above discussion and on the expression obtained in  Equation \eqref{eq:exptykub}, we define the following optimization problem:
\begin{equation}\label{eq:foptg}
\begin{split}
        \fopt \defeq
     \max_{x}& \quad \sum_{k \in \otbuck}\min \left \{\frac{2}{\qr}\sum_{i=1}^{m}\frac{\wl \xik}{1+\sum_{k' \in \otbuck}\wlp \xikp},|\Sk| \right \}\\
     \text{s.t.}& ~0 \leq \xik \leq |\Sk| \text{ for all }i\in \otcust \text{ and }k\in \otbuck~.
     \end{split}
\end{equation}
 The optimal value of the above optimization problem upper bounds the value of the expected number of matches  achieved by the optimal menu. This is formalized in the following lemma, which is formally proved in  Appendix~\ref{appendix:lp_relaxation_general}.
    
    \begin{restatable}[$\fopt$ upper bounds $\opt$]{lemma}{lemfoptg}
    \label{lem:foptg}
        $\opt \leq  \fopt~.$ 
    \end{restatable}
    
      Next, we provide an LP formulation that approximates  optimization problem \eqref{eq:foptg} and therefore, up to constant factors,  upper bounds the value of expected number of matches achieved by the optimal menu.
\begin{equation}\label{eq:lpoptg}
\begin{split}
     \lpopt \defeq \max_{x} & ~\sum_{k \in \otbuck} \frac{2}{\qr}\sum_{i=1}^{m}\wl \xik \\
     \text{s.t.}& ~ \sum_{k' \in \otbuck}\wlp \xikp \leq 1 \text{ for all }i\in \otcust,\\ 
     &~\frac{2}{\qr}\sum_{i=1}^{m}\wl \xik \leq |\Sk| \text{ for all }k \in \otbuck \\
     &~0 \leq \xik \leq |\Sk|\text{ for all }i\in \otcust \text{ and }k\in \otbuck ~.
\end{split}
\end{equation}

\begin{restatable}[$\lpopt$ approximates $\fopt$]{lemma}{finaloptg}
\label{lem:lpoptg}
      $\frac{1}{2}\fopt \leq \lpopt \leq 2 \fopt.$
    \end{restatable}
    
    The proof of \Cref{lem:lpoptg} can be found in  Appendix \ref{appendix:lp_relaxation_general}.  \Cref{lem:foptg} and \Cref{lem:lpoptg} together establish that the the optimal fractional solution to the optimization problem \eqref{eq:lpoptg} with objective value $\lpopt$ upper bounds $\opt$:
    \begin{cor}[{$\lpopt$ upper bounds $\opt$}]\label{cor:lpupg}
    $\opt \leq 2\lpopt.$
    \end{cor}
 
\subsection{Rounding, menu characterization, and constant-factor approximation guarantee}
\label{rounding_general}

Next, we provide the rounding algorithm for the fractional solution returned by problem  \eqref{eq:lpoptg}. Later, using the output of the rounding algorithm, we provide an algorithm to construct the menu for each customer. Finally, we argue that the expected number of matches achieved by the  constructed menu yields a constant approximation to the optimal menu. We now provide the description of our rounding algorithm.


\begin{algorithm}[H]\label{alg:rounding}
	\SetAlgoNoLine
	\KwIn{The optimal (fractional) solution to problem  \eqref{eq:lpoptg}, $\xstar$.}
	\KwOut{An integral solution $x$ that satisfies the properties stated in \Cref{lem:gone}. }
	\BlankLine
	{\bf Algorithm}
	
	For all $i\in \otcust$ and $k \in \otbuck$, assign	$\xik=\floor{\xstar_{i,k}}~\text{if}~ \xstar_{i,k} \geq 1$
	
	For all $\ell \in \Bo$ and $i \in \otcust$, initialize $\yil=0$.
		
		\For{$\ell \in \Bo$}{
		\For{$k \in \otbuck \text{ and }\wl=\ell$}{
		Let $\sk=\sum_{\{i\in \otcust:~\xstar_{ik}<1 \}} \xstar_{ik}$.
		Pick $\ceil{\sk}$ customers with minimum $\yil$ values (breaking ties arbitrarily) in the set $\{ i\in \otcust:~\xstar_{ik}<1\}$; for each such costumer $i$, set $\xik=1$  and increment $\yil=\yil+1$.}
		}
	\caption{Rounding}
	\label{alg:gone}
\end{algorithm}
    \newcommand{\cm}{2c}
    \newcommand{\cmpo}{2c+1}
    
We next characterize some properties of the output of the rounding algorithm. 
    \begin{restatable}{lemma}{rounding}
    \label{lem:gone}
    Algorithm \ref{alg:gone} returns an integral solution $x$ that satisfies the following properties:
    \begin{enumerate}
        \item $  \sum_{k \in \otbuck}\min \left\lbrace \frac{2}{\qr}\sum_{i \in \otcust} \wl \xik,|\Sk|\right\rbrace\geq \frac{1}{2}\lpopt$.
        \item There exists a constant $c\geq 1$ such that $\sum_{k' \in \otbuck}\wlp \xikp \leq c$ for all customers $i\in \otcust$.
        \item $\frac{2}{\qr}\sum_{i=1}^{m}\wl \xik \leq |\Sk|+\frac{2\wl}{\qr}$ for all $k\in \otbuck$.
    \end{enumerate}
    \end{restatable}
   The proof of \Cref{lem:gone} can be found in Appendix \ref{appendix:rounding_menu_general}. The first condition in the lemma implies that the rounded solution $x$ has an objective value at least $\frac{1}{2}\lpopt$ and therefore, up to constant factors, is still an upper bound to the expected number of matches  achieved by the optimal menu. The second condition gives a constant factor upper bound  for the expression $\sum_{k' \in \otbuck}\wlp \xikp$. This inequality is used later to construct a menu set such that the mass of the menu assigned to each customer $i \in \otcust$ (i.e., $\sum_{j\in \mathrm{M}^{i}} \vj$) of each customer is upper bounded by such constant. This will be crucial to achieve the constant factor approximation to our problem (See proof of \Cref{thm:caseA}). The third and the final property is useful in proving the guarantees achieved by the menu construction algorithm presented next (See proof of \Cref{lem:gmenu} for further details). 
   
   We now focus on constructing a menu set using the integral solution returned by the rounding algorithm. In particular, the $\xik$ variables are sufficient to upper bound the value of the expected number of matches for the optimal menu set. We still need to construct the actual menu of suppliers that are going to be shown to each customer. Next, we provide an algorithm to construct the menu for each customer and prove that the menu set constructed by our algorithm achieves a constant approximation to the optimal menu set. 

\begin{algorithm}[H]\label{alg:menu}
	\SetAlgoNoLine
	\KwIn{The assignment $x$ returned by the rounding Algorithm \ref{alg:gone}.}
	\KwOut{A menu set $\{\Mi{x}\}_{i\in \otcust}$ that it satisfies the properties in \Cref{lem:gmenu}.}
	 Set $\ckj=0$ for all buckets $k$ and for all $j \in \Sk$.
	
	\For{$k \in \otbuck$ 
		}{
		\For{$i = 1\dots \cust$
		}{
	Pick $\xik$ suppliers from bucket $k$ with minimum $\ckj$ values (breaking ties arbitrarily), increment their $\ckj$ values by 1, and add these suppliers to menu $\Mi{x}$.
		}
		}
	\caption{Menu construction algorithm}
	\label{alg:gmenu}
\end{algorithm}

In the above algorithm, for each $j \in \Sk$, the final value of $\ckj$ is equal to the number of customers $i$ that satisfy $j \in \Mi{x}$. That is, $\ckj$ tracks to many consumers $j$ in bucket $k$ is shown. Since all customers are symmetric, we show next that it suffices to use the values of the  $\ckj$s to calculate the expected number of matches. 

Our next lemma provides some crucial properties of the $\ckj$ values returned by the Algorithm~\ref{alg:menu}, the menu construction algorithm. 

\begin{restatable}{lemma}{menuconstruction}\label{lem:gmenu}
The menu $\Mi{x}$ for all $i \in \otcust$ returned by Algorithm \ref{alg:gmenu} 
satisfies the following properties  
for each bucket $k \in \otbuck$ 
\begin{enumerate}
    \item $\sum_{j \in \Sk}\ckj= \sum_{i \in \otcust}\xik$.
    \item For all $j\in \Sk$, $\ckj \leq 2+\frac{\qr}{2\wl}$.
\end{enumerate}
\end{restatable}
The proof of \Cref{lem:gmenu} can be found in Appendix \ref{appendix:rounding_menu_general}. As the lemma establishes, our algorithm has the intuitive property that the number of menus in which a supplier appears is upper bounded by a constant that depends both on her (representative) value and the (representative) value of her outside option. In particular, suppliers with lower values can be shown more than those with the same outside option and higher values; this is because suppliers with lower values are less likely to be selected by consumers. On the other hand, comparing suppliers with the same value, those with higher outside options can be shown more often; this implies that, in expectation, such a supplier will be selected by more consumers, which is helpful increases the chances that she will choose one of them to match in the platform. 

 We use these properties to establish the following theorem, which is our main result for this section:

\thmcaseone*
 \begin{proof}
\newcommand{\mukj}{\mu_{k,j}}

  
  Let $x$ be the assignment returned by Algorithm \ref{alg:gone}, and let $\{\Mi{x}\}_{i\in \otcust}$ be the menus generated by Algorithm \ref{alg:gmenu}. Given $\{\Mi{x}\}_{i\in \otcust}$, let $\Xk$ be the random variable counting the number of customers that choose a supplier from bucket $k$, and let $\Xkj$ be the random variable counting the number of customer that choose supplier $j$ in bucket $k$. 
  Let $\Yk$ be the random variable counting the number of suppliers from bucket $k$ that are matched. Also, let $\Ykj{j}$ be the random variable representing whether supplier $j$ in bucket $k$ is matched. Clearly, $\Yk=\sum_{j \in \Sk}\Ykj{j}$ and by taking expectation on both sides, we have, 
  
    \begin{equation}\label{eq:dist}
    \expt{\Yk}=\sum_{j \in \Sk}\expt{\Ykj{j}}.
    \end{equation}
    Now fix a bucket $k \in \otbuck$ and a supplier $j \in \Sk$, the probability of $j$ matching given $\Xkj$ is:
    \begin{equation}\label{eq:prob11}
        \begin{split}
    \prob{\Ykj{j}=1|\Xkj} = \frac{\Xkj}{\qj+\Xkj} \geq \frac{1}{2} \frac{\Xkj}{\qr+\Xkj},
        \end{split}
    \end{equation}
where  the last inequality follows since $\qj \leq 2\qr$.
    
  Define $\mukj\defeq\expt{\Xkj}$ and $\massi\defeq \sum_{j \in \Mi{x}}\vj$ for all $i \in \otcust$. Then, there exists a constant $c$ such that 
\begin{equation}\label{eq:massiup}
    \massi=\sum_{k\in \otbuck} \sum_{\{j \in \Sk: j\in \Mi{x}\}}\vj \leq 2 \sum_{k\in \otbuck}\sum_{j \in \Mi{x}}\wl =2 \sum_{k\in \otbuck}\wl\sum_{j \in \Mi{x}}1 = 2 \sum_{k\in \otbuck}\wl \xik \leq \cm,
\end{equation}
where the last inequality follows from Lemma \ref{lem:gone}. 
We know that, $$\expt{\Xkj} = \sum_{\{i\in \otcust:~j \in \Mi{x}\}}\frac{\vj}{1+\massi}~.$$ 
    We can upper and lower bound the above expression as follows:
\begin{equation}\label{eq:mukj}
    \frac{\wl\ckj}{\cmpo} \leq \mukj=\expt{\Xkj} \leq 2\wl\ckj,
\end{equation}
where we used $\wl \leq \vj \leq 2\wl$, $\ckj=\sum_{\{i\in \otcust:~j \in \Mi{x}\}}1$ (as in defined  Algorithm~\ref{alg:gmenu}) and $\massi \leq \cm$ for all $i \in \otcust$ (\Cref{eq:massiup}).
    
For all buckets $k \in \otbuck$, we divide our analysis into two separate cases, depending on the type of supplier $j \in \Sk$.\\

\noindent {\bf Case 1}: Supplier $j \in \Sk$ satisfies $\wl\ckj>1$. In this case, 
\begin{equation}\label{eq:prob1}
    \prob{\Xkj \leq \frac{\wl \ckj}{2(\cmpo)}}\leq \prob{\Xkj\leq \left(1-\frac{1}{2}\right)\mukj}  \leq \expo{-\frac{\mukj}{8}}  \leq \expo{-\frac{1}{8(\cmpo)}} \leq 1-\frac{e-1}{8e(\cmpo)}.
\end{equation}
The first inequality follows by Equation \eqref{eq:mukj}, and the second inequality follows by the Chernoff bound. The third inequality follows by combining Equation \eqref{eq:mukj} with $\wl\ckj > 1$, and the final inequality follows by the known fact that, for any $x \in [0,1]$ we have $e^{-x} \leq 1-\frac{e-1}{e}x$ (see \Cref{fact:up} in Appendix~\ref{app:tools}). 

We can now calculate the expected value of $\Ykj{j}$.
    \begin{equation}\label{case1}
        \begin{split}
    \expt{\Ykj{j}}&=\prob{\Ykj{j}=1}=\sum_{a}\prob{\Ykj{j}=1|\Xkj = a}\prob{\Xkj=a } \geq \frac{1}{2} \sum_{a}\frac{a}{\qr+a}\prob{\Xkj=a}\\
    & \geq\frac{1}{2}\frac{1}{2(\cmpo)} \frac{\wl \ckj}{\qr+\wl\ckj}\prob{\Xkj\geq \frac{\wl \ckj}{4(\cmpo)} } \geq \frac{1}{2(\cmpo)} \frac{\wl \ckj}{4\qr} \frac{e-1}{8e(\cmpo)}~.
        \end{split}
    \end{equation}
  The first inequality follows from \eqref{eq:prob11}. For the second inequality, we consider only those $a$ that satisfy $a \geq \frac{1}{2(\cmpo)} \wl \ckj$, and for all such $a$ we have $\frac{a}{\qr+a} \geq \frac{1}{2(\cmpo)} \frac{\wl \ckj}{\qr+\wl\ckj}$. In the final inequality we use Equation \eqref{eq:prob1} to lower bound the probability term, and use the fact that $\ckj \leq 2+\frac{\qr}{2\wl}$ (Lemma~\ref{lem:gmenu}), which further implies that $\wl \ckj \leq 2\wl+\frac{\qr}{2} \leq 3\qr$.\\

\noindent {\bf Case 2}: Supplier $j \in \Sk$ satisfies $\wl\ckj\leq 1$. In this case,
\begin{equation}
\label{eq:x_k}
\begin{split}
    \prob{\Xkj =0}&=\prod_{\{i\in \otcust:~j \in \Mi{x}\}}\prob{i \text{ does not pick } j}=\prod_{\{i\in \otcust:~j \in \Mi{x}\}}\left(1-\frac{\vj}{1+\alpha_{i}}\right) \\
    &\leq  \left(1-\frac{\wl}{\cmpo}\right)^{\ckj}\leq \expo{-\frac{\wl\ckj}{\cmpo}}\leq 1-\frac{e-1}{e}\frac{\wl\ckj}{\cmpo}.
    \end{split}    
\end{equation}
The second equality follows from the fact that, for all $i$ such that $j \in \Mi{x}$, the probability that $i$ picks $j \in \Sk$ is equal to $\frac{\vj}{1+\massi}$.  Further, we have that $\frac{\vj}{1+\massi}$ is greater than $\wl/(\cmpo)$ because $\massi \leq \cm$ (\Cref{eq:massiup}). Further, there are $\ckj$ customers $i$ such that $j \in \Mi{x}$, which gives us the first inequality. The final inequality follows from \Cref{fact:lb} and \Cref{fact:up} respectively in Appendix~\ref{app:tools}. 
We can now bound $\expt{\Ykj{j}}$ as follows 
\begin{equation}\label{eq:case2}
        \begin{split}
    \expt{\Ykj{j}}=\prob{\Ykj{j}=1}&=\sum_{a}\prob{\Ykj{j}=1|\Xkj = a}\prob{\Xkj=a } \geq \frac{1}{2} \sum_{a}\frac{a}{\qr+a}\prob{\Xkj=a}\\
    &\geq \frac{1}{2}\frac{1}{\qr+1}\prob{\Xkj\geq 1} \geq \frac{1}{4\qr}\frac{e-1}{e}\frac{\wl\ckj}{\cmpo}.
        \end{split}
    \end{equation}
    The first inequality follows from Equation \eqref{eq:prob11} and the second inequality follows because $\frac{a}{\qr+a} \geq \frac{1}{\qr+1}$ for all $a \geq 1$. The final inequality follows from \Cref{eq:x_k}.
    
\vspace{0.5cm}

    Now combining the two cases together for all buckets $k \in \otbuck$, we have,
        \begin{equation}\label{eq:dist}
        \begin{split}
            \expt{\Yk}=\sum_{j \in \Sk}\expt{\Ykj{j}}\geq \sum_{j \in \Sk}C \frac{\wl\ckj}{\qr}= \sum_{i}C \frac{\wl\xik}{{\qr}}
        \end{split}
    \end{equation}
for constant $C = \frac{e-1}{64e(\cmpo)^2}$. The inequality follows by combining Equations \eqref{eq:case2} and \eqref{case1} together and the final equality follows from  Lemma \ref{lem:gmenu}. Now, given $\Mi{x}$ for all $i \in \otcust$, let $\textbf{M}$ be the random variable counting the number of matched suppliers. We have $\textbf{M}=\sum_{k\in \otbuck}\Yk$, and taking expectation on both sides and further lower bounding this quantity we get,
    \begin{align}
        \expt{\textbf{M}}=\sum_{k}\expt{\Yk} \geq C\sum_{k}\sum_{i} \frac{\wl\xik}{\qr} \geq \frac{C}{2} \lpoptg
    \end{align}
    and we have a $c^{\prime \prime} =\frac{C}{2}$ approximation to the optimum menu via Lemma \ref{lem:lpoptg}.
\end{proof}

\section{Proof overview of the high value supplier's case}
\label{subsec:largeSupplier}

In this section, we focus on the case where all suppliers $j\in \otsupp$ have scores $\vj \geq 1$ and provide an efficient algorithm that produces menu set which is a constant factor approximation to the expected number of matches  achieved by the optimal menu.

Similar to \Cref{subsec:smallSupplierLargeOutside} we start by providing an upper bound to the expected number of matches  achieved by the optimal menu. Our upper bound,  described in the following lemma, is different from the one in the previous section. It is worth highlighting, howver, that the upper bound does not make any assumptions on the $v_j$s ans thus it is also valid for the low-value suppliers' setting described before. 

\begin{restatable}{lemma}{lemcasetwoup}\label{lem:casetwoup}
Let $\cust$ and $\supp$ be the number of customers and suppliers, respectively. 
Then, the expected number of matches achieved by the optimal menu ($\opt$) is upper bounded by,
\begin{equation}\label{eq:case1ub}
    \max_{\{x \in \Z_{\geq 0}^{\supp}:~\sum_{j=1}^{\supp}\xj=\cust\}}\sum_{j=1}^{\supp}\frac{\xj}{\xj+\qj}
\end{equation}
\end{restatable}

The proof of \Cref{lem:casetwoup} can be found in \Cref{app:casetwo}. 
Intuitively, imagine an alternative setting where the platform can present one option to each customer and she will choose this option with probability one (equivalently, suppose that customers do not have outside options). In that setting, the optimal strategy for the platform is to present each customer with a single supplier; this way, the platform avoids the risk of coordination failures, leading to many customers choosing the same supplier. Then, optimization problem in \eqref{eq:case1ub} essentially provides the optimal solution to the platform's problem in this alternate setting. Moreover, as the platform faces no uncertainty on the customer side in the alternate setting, its solution provides an upper bound to our original problem.


Note that the above problem is a discrete problem. 
The following lemma shows that there exist an efficient algorithm to solve the optimization problem \eqref{eq:case1ub} up to a multiplicative $1/2$ approximation.

\begin{restatable}{lemma}{lemeff}\label{lem:eff}
There exists an efficient algorithm whose output is an integral vector $y\in \Z_{\geq 0}^{\supp}$  that satisfies $\sum_{j=1}^{\supp}y_j=\cust$ and 
\begin{equation}\label{eq:case2ub}
\sum_{j=1}^{\supp}\frac{y_j}{y_j+\qj}\geq \frac{1}{2}    \max_{\{x \in \Z_{\geq 0}^{\supp}:~\sum_{j=1}^{\supp}\xj=\cust\}}\sum_{j=1}^{\supp}\frac{\xj}{\xj+\qj}.
\end{equation}
\end{restatable}
The proof can be found in \Cref{app:casetwo}. 
In what follows, using the solution $y$ identified in the above lemma to construct a menus set  that matches the upper bound in \eqref{eq:case1ub} up to a constant factor. We summarize this result in the following theorem. 

\renewcommand{\Xj}{\textbf{X}_{j}}
\renewcommand{\Yj}{\textbf{Y}_{j}}

\thmcasetwo*
\begin{proof}
\renewcommand{\xstar}{y^{*}}
Let $\xstar$ be a $\frac{1}{2}$-approximate solution to the optimization problem in \eqref{eq:case1ub}. Recall that, by Lemma \ref{lem:eff}, such a solution exists and can be computed efficiently. Now, by Lemma~\ref{lem:casetwoup}, we have that 
$$\frac{1}{2}\opt \leq \sum_{j=1}^{\supp}\frac{\xstar_{j}}{\xstar_{j}+\qj}.$$

Next, we show a construction of menus for customers that approximates this upper bound up to a constant factor.

Consider the following construction of menus $\textbf{M}_{i}$ for all $i \in \otcust$. Set $\textbf{M}_{i}=\{j \}$ for some $\xstar_{j}$ number of customers. Since $\sum_{j=1}^{\supp}\xstar_{j}=\cust$, this is a valid menu, i.e., we can find a non-overlapping partition of customers such that each partition is shown the same supplier $j$ (and just $j$).

Now lets calculate the expected number of matches given these menus. Given menus $\{\textbf{M}_{i} \}_{i \in \otcust}$, let $\Xj$ for all $j \in \otsupp$ be a random variable counting the number of customers that selected  supplier $j$. Further, let $\Yj$ be the random variable denoting whether  supplier $j$ is matched  or not. Then,
\begin{equation}\label{eq:algexpt}
    \begin{split}
\expt{\Yj}& =\expt{\expt{\Yj|\Xj}}=\expt{\frac{\Xj}{\Xj+\qj}}=\sum_{a =0}^{\xstar_j}\frac{a}{a+\qj}\prob{\Xj=a}\\
&\geq  \sum_{a\geq \xstar_{j}/4}\frac{1}{4}\frac{\xstar_{j}}{\xstar_{j}+\qj}\prob{\Xj=a} \geq \frac{1}{4}\frac{\xstar_{j}}{\xstar_{j}+\qj}\prob{\Xj\geq \xstar_{j}/4}.
    \end{split}
\end{equation}

To lower bound the quantity $\prob{\Xj\geq \xstar_{j}/4}$, first observe that
\begin{equation}\label{eq:exptcase1}
\expt{\Xj}=\sum_{\{i \in \otcust:~j \in \textbf{M}_{i}\}}\frac{\vj}{\vj+1}\geq \frac{\xstar_{j}}{2},    
\end{equation}
where in the inequality above we used the fact that $\vj \geq 1$ and there are $\xstar_{j}$ number of customers for which $j \in \textbf{M}_{i}$. 

An upper bound on the quantity $\prob{\Xj \leq \xstar_{j}/4}$ when $\xstar_{j}\geq 1$ can be obtained as follows:
\begin{equation}\label{eq:case1prob}
    \begin{split}
        \prob{\Xj \leq \xstar_{j}/4} \leq \prob{\Xj \leq (1-1/2)\expt{\Xj}} \leq \expo{\frac{-\expt{\Xj}}{12}} \leq \expo{-1/24}
    \end{split}
\end{equation}
For the first inequality, note that by \Cref{eq:exptcase1}, we have that $\xstar_{j} \leq 2 \expt{\Xj}$. The second inequality follows from the standard Chernoff bound, and the third one from the fact that $\expt{\Xj} \geq 1/2$ (by  \Cref{eq:exptcase1} combined with $\xstar_{j} \geq 1$). 

Combining Equations \eqref{eq:case1prob} and \eqref{eq:algexpt}, we get
$$\expt{\Yj} \geq \frac{1}{4}\frac{\xstar_{j}}{\xstar_{j}+\qj}\prob{\Xj\geq \xstar_{j}/4} \geq \frac{1}{4}\frac{\xstar_{j}}{\xstar_{j}+\qj}(1-\expo{-1/24})$$
Therefore, the expected number of matches produced by our algorithm is lower bounded by,
$$\sum_{j=1}^{\supp}\expt{\Yj} \geq \alpha_2\sum_{j=1}^{\supp}\frac{\xstar_{j}}{\xstar_{j}+\qj} \geq \alpha_2 *\opt$$ for constant $\alpha_2=1/8 (1-\expo{-1/24})$. Hence, we have a constant factor approximation for this case.
\end{proof}

\section{Simulations}
\label{subsec:simulations}

In the previous sections we focused on constructing an efficient algorithm with  constant-factor approximation guarantees. 
To gain a better understanding of the performance of the proposed algorithm under perhaps more natural conditions, we provide a series of simulations that compare the expected number of matches achieved by our algorithm to those achieved by an upper bound. 

To test the performance of our algorithm we report the results for markets with $100$ suppliers and vary the number of customers $m$, and the values of the suppliers and their outside options. We focus on the low-value supplier regime, where $v_j \leq 1$ and $q_j \geq 1$ for all suppliers $j \in \otsupp$.
Specifically, we generate the value of the suppliers as $v_j = 1/(1+z_j)$  where the $z_j$'s are i.i.d.~draws from an exponential distribution with parameter $\lambda_v$. Note that, as $\lambda_v$ increases, the value of the suppliers becomes smaller in expectation. The value of the outside options are given by $q_j = 1+w_j$ where the $w_j$'s are i.i.d.~draws from an exponential distribution with parameter $\lambda_o$.\footnote{To test the robustness of the reported numbers, we ran the simulations with varying number of suppliers and using different distributions to determine the values of  the $v_j$s and $q_j$s. Results are qualitatively the same, and are omitted for the sake of brevity. }

For each combination of $(m,\lambda_v,\lambda_o)$, we  generate 25 instances, where each instance is constructed by drawing a set of $v_j$s and $q_j$s according to the given distributions. For each instance, we compute the upper bound as well as the menu set prescribed by our algorithm. To calculate the expected number of matches resulting from our algorithm, for each instance we run 30 simulations of the choices and report the average number of matches obtained. Results are summarized in Table~\ref{tab:simul_results}.

A major challenge in testing the performance of our algorithm is the difficulty in producing a tight upper bound in for the low-value supplier setting. For simplicity, we will use the \textit{linear relaxation} of the upper bound provided in Lemma~\ref{lem:casetwoup}. However, as discussed in Section~\ref{subsec:largeSupplier}, the aforementioned bound works as an upper bound to the problem even in a setting where costumers do not have outside options. Recall that, in these simulations, we are assuming that costumers have an outside option whose value has been normalized to $1$, while the value of each supplier is at most one. Therefore, we do expect to see some loss associated to using this upper bound.

Surprisingly, we observe that our algorithm is able to perform well even against this upper bound. In particular, we observe that at least 1/3 of the profit is consistently achieved throughout the test instances. 

{
\begin{table}[]
\begin{center}
\small{
\begin{tabular}{ | c || c | c || c | c || c | c | c | }  \hline
m	&	$\lambda_v$	&	$\lambda_o$	&	Avg[ALG]	&	Avg[UB]	&	Mean	&	Min	&	Median	\\ \hline \hline
\multirow{4}{*}{50}	&	1	&	1	&	10.63	&	23.50	&	0.45	&	0.43	&	0.45	\\ \cline{2-8}
	&	1	&	10	&	5.76	&	12.17	&	0.47	&	0.42	&	0.47	\\ \cline{2-8}
	&	10	&	1	&	9.67	&	23.78	&	0.41	&	0.38	&	0.41	\\\cline{2-8}
	&	10	&	10	&	5.44	&	12.47	&	0.44	&	0.40	&	0.44	\\ \hline \hline
\multirow{4}{*}{75}	&	1	&	1	&	13.73	&	30.88	&	0.44	&	0.42	&	0.44	\\ \cline{2-8}
	&	1	&	10	&	7.50	&	15.91	&	0.47	&	0.44	&	0.47	\\ \cline{2-8}
	&	10	&	1	&	12.27	&	30.67	&	0.40	&	0.37	&	0.40	\\ \cline{2-8}
	&	10	&	10	&	7.00	&	15.64	&	0.45	&	0.39	&	0.45	\\ \hline \hline
\multirow{4}{*}{100}	&	1	&	1	&	16.04	&	36.74	&	0.44	&	0.41	&	0.44	\\ \cline{2-8}
	&	1	&	10	&	8.83	&	18.97	&	0.47	&	0.43	&	0.47	\\ \cline{2-8}
	&	10	&	1	&	14.00	&	36.63	&	0.38	&	0.35	&	0.38	\\ \cline{2-8}
 &	10	&	10	&	8.21	&	18.87	&	0.44	&	0.40	&	0.43	\\ \hline \hline
\multirow{4}{*}{125} 	&	1	&	1	&	17.22	&	41.40	&	0.42	&	0.38	&	0.42	\\ \cline{2-8}
	&	1	&	10	&	9.78	&	20.77	&	0.47	&	0.42	&	0.48	\\ \cline{2-8}
	&	10	&	1	&	15.57	&	41.37	&	0.38	&	0.35	&	0.38	\\ \cline{2-8}
	&	10	&	10	&	9.55	&	21.29	&	0.45	&	0.43	&	0.45	\\ \hline \hline
\multirow{4}{*}{150}	&	1	&	1	&	18.48	&	45.98	&	0.40	&	0.38	&	0.40	\\ \hline
	&	1	&	10	&	10.90	&	23.38	&	0.47	&	0.42	&	0.47	\\ \cline{2-8}
	&	10	&	1	&	16.75	&	45.72	&	0.37	&	0.33	&	0.37	\\ \cline{2-8}
	&	10	&	10	&	10.17	&	23.30	&	0.44	&	0.41	&	0.44	\\ \hline \hline
\multirow{4}{*}{200}	&	1	&	1	&	20.49	&	52.36	&	0.39	&	0.37	&	0.39	\\ \cline{2-8}
	&	1	&	10	&	12.49	&	27.29	&	0.46	&	0.41	&	0.45	\\ \cline{2-8}
	&	10	&	1	&	18.83	&	52.71	&	0.36	&	0.34	&	0.36	\\ \cline{2-8}
	&	10	&	10	&	11.93	&	27.44	&	0.44	&	0.37	&	0.44	\\ \hline

\end{tabular}
}
\end{center}\caption{
Simulation results for markets with $100$ suppliers, and $m$ customers. For each combination of $(m,\lambda_v,\lambda_o)$, we generated 25 instances, where in each instance the value of the suppliers is $v_j = 1/(1+z_j)$  where the $z_j$'s are i.i.d.~draws from an exponential distribution with parameter $\lambda_v$; the value of the outside options are given by $q_j = 1+w_j$ where the $w_j$'s are i.i.d.~draws from an exponential distribution with parameter $\lambda_o$.  $Avg[ALG]$ reports the average number of matches achieved by our algorithm across instances and simulations; $Avg[UB]$ reports the average of the upper bound across instances. For each instance, we calculate $Avg[\textup{matches in algorithm}]/UB$; Mean, Min, and Median report the corresponding quantities across the 25 instances.}
    \label{tab:simul_results}
\end{table}
}

\section{Conclusion and open questions}
\label{sec:conclusion}

Platforms for two-sided markets face the challenging task of  providing their users with match recommendations. Taking an algorithmic approach to this challenge, this paper introduced a stylized two-stage model for assortment planning in two-sided matching markets. Agents in the model have public scores as well as an outside option score and agents' choices from given menus follow a distribution proportionate to the relevant scores.
The key contribution is an algorithm that construct menus that are shown to customers that provides a constant approximation algorithm to the optimal expected number of matches.


Several open questions follow directly from this work. The first  question is to improve the upper bound and further provide a non-trivial lower bound. Second, one could try to formally address the case in which   customers are indeed different. We believe that the same techniques will provide a constant-factor approximation algorithm if customers have a publicly known score. Third, is to allow agents to have different types, thus having agents disagree over public scores. We believe that, as long as there is a constant fraction of each type of agents, there is an efficient algorithm that yields a constant-factor  approximation.

Other directions we find intriguing  are the following. First, the platform could determine not only the menus but also  which set of agents should choose first, a decision that has been proved to greatly impact the outcome in other settings (see, e.g.,  \citet{kanoria2017facilitating}). Second, the platform could allow for a constant number of substitutions, thus allowing a ``rejected" agent to choose another partner from her menu; see \citet{liu2019competing} who take a multi-arm bandit approach to a similar problem. Finally, it would be interesting to extend this setting to allow for  more general choice models. 
\bibliographystyle{plainnat}
\bibliography{references}
\appendix
\newpage
\section{Useful Machinery for the Proofs}\label{app:tools}
First we state some standard inequalities which we shall use throughout the proof section.

\begin{fact}[{Jensen's Inequality}]\label{fact:jensens}
If $X$ is a random variable and $\fnf$ is any concave function then, $\expt{\fnf(X)} \leq \fnf(\expt{X})$.
\end{fact}

\begin{fact}\label{fact:lb}
For any $x \in \R$, the following holds: $1-x \leq e^{-x}$.
\end{fact}

\begin{fact}\label{fact:up}
For any $x \in [0,1]_{\R}$, the following holds: $e^{-x} \leq 1-\frac{e-1}{e}x$.
\end{fact}

\begin{fact}[{Chernoff bound}]\label{fact:chernoff}
Suppose $X_1, \dots , X_n$ are independent random variables taking values in $\{0, 1\}$. Let $X=\sum_{i=1}^{n}X_i$ and let $\mu = \expt{X}$. Then for any $\delta \in [0,1]_{\R}$,
$$\prob{X \geq (1+\delta)\mu} \leq \expo{-\frac{\mu \delta^2}{3}}\quad \quad \prob{X \leq (1-\delta)\mu} \leq \expo{-\frac{\mu \delta^2}{2}}~.$$
\end{fact}

\noindent We state and prove a technical lemma that lower bounds the expected number of distinct elements observed after $n$ draws from a distribution close to uniform. 

\begin{restatable}[{Distinct elements lemma}]{lemma}{lemdistinct}
\label{lem:distinct}
Let $p=(p_1,\dots p_n)$ be a discrete distribution over $n$ elements that satisfies $p_i \leq c * p_j$ for all $i,j \in [1,n]_{\Z}$ and some fixed $c>1$. Given $m$ i.i.d samples from distribution $p$, the expected number of distinct elements observed is at least,
$$\dconst\frac{\min(m,n)}{c}~.$$
\end{restatable}

\begin{proof}
Without loss of generality assume $p_1 \leq p_2 \dots \leq p_{n}$. Note that $p_{n}\geq \frac{1}{n*c}$, else we know that $p_{i} \leq c*p_{n}< 1/n$ and $\sum_{i}p_{i}<1$, which is a contradiction. Hence $p_{n}\geq \frac{1}{n*c}$ which further implies $p_{i}\geq \frac{1}{n*c}$ for all $i=1\dots n$.

$\bullet$ $m \leq n$. Let $X_i$ be a binary random variable that is $1$ if the $i^{th}$ element appears at least once in the $m$ samples.  We know that the probability $\prob{X_i=0} = (1-p_i)^m$, and therefore $\expt{X_i} = 1 - (1-p_i)^m$. Now the expected number of distinct elements seen after $m$ samples from distribution $p$ is given by:
\begin{align*}
    \expt{\sum_{i=1}^{n}X_i} &= \sum_{i=1}^{n}\expt{X_i} = \sum_{i=1}^{n}(1-(1-p_{i})^{m}) \geq \sum_{i=1}^{n}(1-(1-\frac{1}{n*c})^{m}) \\
&\geq \sum_{i=1}^{n}1-e^{-m/(n*c)}\geq \sum_{i=1}^{n}(1-(1-\dconst\frac{m}{n*c})) \geq \dconst\frac{m}{c},
\end{align*}
where the second last uses a standard inequality. \footnote{$a^{-x} \leq \frac{a-1}{a}x$ for all $x \in [0,1], a \geq 1$.}

$\bullet$ $m > n$. Since the expected number of distinct elements seen after $m$ samples is monotonically increasing in $m$, we can lower bound this case by assuming $m=n$. More formally, the expected number of distinct elements seen after $m$ samples is
$$\sum_{i=1}^{n}(1-(1-p_{i})^{m}) \geq \sum_{i=1}^{n}(1-(1-\frac{1}{n*c})^{n}) \geq \dconst\frac{n}{c}.$$

Combining the two cases above gives us the proof.
\end{proof}

\newcommand{\Xpi}{\textbf{X}'_{i}}
\newcommand{\rvZ}{\textbf{Z}}
\newcommand{\rvX}{\textbf{X}}
\newcommand{\rvXp}{\textbf{X}'}
\begin{lemma}[\citet{pomatto2018stochastic}]\label{lem:stocdom}
Suppose $\rvX$ stochastically dominates $\rvXp$ and $\rvZ$ is independent of $\rvX$ and $\rvXp$, then $\rvX+\rvZ$ stochastically dominates $\rvXp+\rvZ$
\end{lemma}
\newpage

\section{Proofs omitted in Section \ref{sec:results}}
\label{appendix_results}
\newtheorem{reduction}{Reduction}
\newcommand{\va}{\textbf{a}}
\newcommand{\Sstar}{\textbf{S}^{*}}

\subsection{Proof of Proposition~\ref{prop:sequential_np_hard}}
For the sake of completeness, we start by defining the $3$-partition formally below.

\begin{definition}[$3$-partition problem]
Given a set $A = \{a_1, a_2, \ldots, a_\cust\}$ of $\supp = 3\cust$ positive integers, and another positive integer $B$, such that $\frac{B}{4} < a_j < \frac{B}{2}$ for all $a_j \in A$, and such that $\sum_{a_j \in A} a_j= \cust B$; can $A$ be partitioned into $m$ disjoint sets $A_1,A_2, \ldots, A_m$ such that $\sum_{a \in A_i} a = B$ for all $1 \leq i \leq \cust$?
\end{definition}

Note the conditions in the above problem implies that each $A_i$ must contain \textit{exactly} three elements from A. 

\begin{lemma}[\citet{garey2002computers}]
The $3$-partition problem is strongly NP-Complete.
\end{lemma}

We now reduce the above $3$-partition problem into an instance of the  sequential assortment problem.

\begin{reduction} \label{red:hardness}
Our reduction works as follows:

\begin{enumerate}
    \item Let $\supp$ be the number of suppliers, where supplier $j\in \otsupp$ has score $\vj$ defined as follows:
    $$\vj = \frac{a_j}{\sum_{j' \in \otsupp}a_{j'}}$$
    Note that $\sum_{j}\vj=1$, and so we normalize these scores only for convenience; see proof of \cref{prop:sequential_np_hard}. Each supplier has outside option $q_j=0$.
    \item Let $\cust = \frac{n}{3}$ be the number of customers. Each customer has an outside option equal to $v_{0}$. To simplify notation in the subsequent lemmas, we allow $v_0$ to be arbitrary. In the proof of Proposition~\ref{prop:sequential_np_hard} we show how  this is without loss of generality, as we can always re-scale the scores $\vj$s so as to have $v_0=1$. 
\end{enumerate}
\end{reduction}
\newcommand{\match}{\textbf{M}}
\newcommand{\vmin}{v_{\min}}
\newcommand{\vmax}{v_{\max}}
Define $\vmin= \min_{j \in \otsupp }\vj$ and $\vmax=\max_{j\in \otsupp}\vj$.
\begin{lemma}\label{lem:disjoint}
Let $\cust$ and $\supp$ be the number of customers and suppliers. Let $\vj>0$ be the score of supplier $j\in \otsupp$ and suppose that $\sum_{j}\vj=1$. If $v_0 < \frac{\vmin^3}{4k}$, then the optimal menus for customers are disjoint. 
\end{lemma}
\begin{proof}
Denote the set of customers by $\otcust = \{1,2,\ldots,m\}$, and the set of suppliers by $\otsupp = \{1,2,\ldots,n\}$. 
Let the optimal menus be given by $\Mset = \menuset$. We prove the statement by contradiction. To that end, assume, without loss of generality, that there exists a supplier $h \in \otsupp$, that belongs to menus of both customer $1$ and customer $2$, i.e., $h \in \M_{1}\cap \M_{2}$. Let $\match$ be the random variable denoting the number of matches under the optimal menus $\Mset$.

\newcommand{\xoh}{\textbf{X}_{1,h}}
\newcommand{\xth}{\textbf{X}_{2,h}}
\newcommand{\ao}{\alpha_{1}}
\newcommand{\at}{\alpha_{2}}
\newcommand{\vh}{v_{h}}
\newcommand{\vz}{v_{0}}

We next calculate the probability that customers $1$ and $2$ both end up choosing supplier $h$. To that end, let $\xoh, \xth$ denote the indicator variables corresponding to whether customers $1, 2$, respectively, choose supplier $h$.
Let $\ao\defeq \sum_{j\in \M_{1}}\vj$, and $\at\defeq \sum_{j\in \M_{2}}\vj$, and note that, by the definition of the $\vj$s, we have  $\ao,\at \leq 1$. Then, 

$$\prob{\xoh=1;\xth=1}=\frac{\vh^2}{(\ao+\vz)(\at+\vz)} \geq \frac{\vh^2}{(1+\vz)^2}~,$$
where the equality follows from  the fact that consumers' choices are independent. 

Now, whenever customers $1$ and $2$ both choose supplier $h$, we must have a matching of size at most $\cust-1$ as, by definition, we have $\cust < \supp$. Therefore,
$$\prob{\match\leq \cust-1} \geq \prob{\xoh=1; \xth=1} \geq \frac{\vh^2}{(1+\vz)^2}~.$$

Using the above observation to bound $\expt{\match}$, we get
\begin{align}\label{eq:hardness1}
    \begin{split}
        \expt{\match} &\leq \cust \prob{\match=\cust} + (\cust-1) \prob{\match \leq \cust-1}\\
        &=\cust (1-\prob{\match \leq \cust-1}) + (\cust-1)\prob{\match \leq \cust-1}\\
        &=\cust-\prob{\match \leq \cust-1}\\
        &\leq \cust-\frac{\vh^2}{(1+\vz)^2} \leq \cust-\frac{\vmin^2}{(1+\vz)^2}~.
    \end{split}
\end{align}

Consider the following disjoint menus $\M'_{i}$ for customers $i\in \otcust$:
$$\M'_{i}\defeq \{v_{i} \}.$$

In other words, the $i$-th customer just observes $i$-th supplier, and since $\cust\leq \supp$ this is a valid assignment (and all menus are disjoint). We will show that there exists a $\vz$ for which these new disjoint menus ($\{\M^\prime_{i}\}_{i\in \otcust}$), despite disregarding suppliers in the set $[\cust+1,\supp]$, perform strictly better than the optimal menus $\{\M_{i}\}_{i\in \otcust}$ defined earlier, thus reaching a contradiction. 

Let $\match'$ be the expected size of the matching under menus $\{\M^\prime_{i}\}_{i\in \cust}$. Then,
\newcommand{\vi}{v_{i}}
\begin{align}\label{eq:hardness2}
        \expt{\match'}&=\sum_{i=1}^{\cust}\frac{\vi}{\vz+\vi}= \sum_{i=1}^{\cust}\left(1- \frac{\vz}{\vi+\vz}\right) \geq \sum_{i=1}^{\cust} \left(1- \frac{\vz}{\vmin+\vz}\right)=\cust\left(1- \frac{\vz}{\vmin+\vz}\right)~,
\end{align}
where we used the fact that the menus are disjoint and that each supplier has $q_j=0$.

If $\vz < \frac{\vmin^3}{4\cust}$, this implies that note that $0< \frac{\vmin^3}{4\cust} < 1$, and from Equation \eqref{eq:hardness1} above we obtain
$$\expt{\match} \leq \cust-\frac{\vmin^2}{(1+\vz)^2} \leq \cust-\frac{\vmin^2}{4}.$$

And from Equation \eqref{eq:hardness2}, we get
$$\expt{\match'} \geq \cust\left(1- \frac{\vz}{\vmin+\vz}\right) \geq \cust\left(1-\frac{\vz}{\vmin} \right)~.$$

Finally,
$$\vz < \frac{\vmin^3}{4\cust} \implies \cust-\frac{\vmin^2}{4} < \cust\left(1-\frac{\vz}{\vmin} \right) \implies \expt{\match}< \expt{\match'},$$
which is a contradiction to the optimality of $\Mset$
\end{proof}

\begin{lemma}\label{lem:disjoint2} Consider the instance constructed in Reduction~\ref{red:hardness}. The menus $\{\M_{i}\}_{i \in \otcust}$ are said to be balanced if 
\begin{align*}
   \sum_{j\in \M_{i}}\vj = \frac{1}{\cust}, ~\forall i \in \otcust.
\end{align*}
Whenever  the optimal menus are disjoint, and it is feasible to construct balanced menus, 
the optimal menus $\{\M_{i}\}_{i \in \otcust}$ are balanced.
\end{lemma}
\begin{proof}
Let $\{\M_{i}\}_{i \in \otcust}$ be the optimal menus and suppose they are disjoint and that it is feasible to construct balanced menus. Let $\match$ be the expected size of the matching given the optimal menus $\{\M_{i}\}_{i \in \otcust}$   Define $\alpha_{i} = \sum_{j\in \M_{i}}\vj$.

Given that the optimal menus are disjoint, we have $\sum_{i\in \otcust}\alpha_{i}\leq \sum_{j\in \otsupp}\vj \leq 1$. Therefore,
$$\expt{\match} = \sum_{i\in \otcust} \frac{\alpha_{i}}{v_0+\alpha_{i}} \leq \max_{\{x\in \R^\cust~|~\sum_{i\in \otcust}x_{i}\leq 1\}} \sum_{i\in \otcust}\frac{x_{i}}{v_0+x_{i}}=\cust(\frac{1/\cust}{v_0+1/\cust})=\frac{\cust}{mv_0+1}~.$$
The third equality follows because, given any constant $v_0>0$, the function $\sum_{i\in \otcust}\frac{x_{i}}{v_0+x_{i}}$ is concave in $x$, and the maximum under the constraint $\sum_{i\in \otcust}x_{i}\leq 1$ is achieved when $x_{i}=\frac{1}{\cust}$ for all $i\in \otcust$. Note that, if it is feasible to construct balanced menus, then this bound is achieved at equality by any balanced menus. 
\end{proof}

\nphardness*

\begin{proof}
Assume to the contrary that the sequential assortment problem is not strongly NP-hard. This implies that there exists an algorithm that solves the Sequential Assortment problem whose running time is polynomial in the size of the input.

Now, given an instance of the $3$-partition problem, we can reduce it to an instance of the Sequential Assortment problem as in Reduction \ref{red:hardness} above. By way of Lemma \ref{lem:disjoint}, as long as we define $v_0 = \frac{v_{min}^3}{8\cust} \leq \frac{v_{min}^3}{4\cust}$, the optimal menus will be disjoint. To be consistent with our definition of the Sequential Assortment in Section \ref{sec:model}, we can re-scale $v_0$ to $1$, and all other values $v_j$ by $\frac{8\cust}{v_{min}^3}$, and by the scale invariance of the MNL model, the optimal solution to the problem remains the same. Note that this value $\frac{8\cust}{v_{min}^3}$ is polynomial in the size of the input of the $3$-partition instance.

It is easy to see, using Lemma \ref{lem:disjoint2}, that whenever an instance of the $3$-partition problem has a feasible partition, the optimal menus of the corresponding instance of the Sequential Assortment problem are balanced. 

Therefore, a weakly polynomial time algorithm for the Sequential Assortment problem is enough to solve the $3$-partition problem in wealy polynomial time.
\end{proof}

\renewcommand{\Xij}{\textbf{X}_{i,j}}

\subsection{Proof of Theorem~\ref{thm:combine}}

\thmcombine*
\begin{proof}
Define $OPT(t,X)$ as 
the maximum possible expected number of matches in a market with $t$ customers and a set $X \subseteq{\otsupp}$ of suppliers with their associated public scores $\{v_j\}_{j\in X}$ and  outside option scores $\{q_j\}_{j\in X}$. 

Let $A$  be the set of suppliers with scores $ v_j \geq 1$ and $B$  be the set of suppliers with scores $ v_j < 1$. Hence we have $A \cup B = \otsupp$ and $A \cap B = \emptyset.$

For ease of exposition, assume that $m$, the number of customers,  is even; the proof can be trivially adapted if not.
Suppose that the following inequality holds
\begin{align}\label{eqn:claim}
    OPT(m, A \cup B ) \leq 2 \left[OPT\left(\frac{m}{2},A\right) + OPT\left(\frac{m}{2},B\right) \right].
\end{align}

By $\Cref{thm:caseA}$ and $\Cref{thm:caseB}$, we have algorithms $\mathbb{A}_L$ and $\mathbb{A}_H$  that achieve a constant approximation ratio of $\alpha_L$ and 
$\alpha_H$ for the low and high-value supplier cases, respectively. Let $\mathbf{M}_{\mathbb{A}_L}$ denote the number of matches achieved by algorithm $\mathbb{A}_L$ in a market with $m/2$ customers and the set of suppliers given by $A$ and let  $\mathbf{M}_{\mathbb{A}_H}$ denote the number of matches achieved by algorithm $\mathbb{A}_H$ in a market with $m/2$ customers and the set of suppliers given by $B$. 
Then, we have that 
$\expt{\mathbf{M}_{\mathbb{A}_L}} \geq \alpha_L OPT(\frac{m}{2},A) $ and $\expt{\mathbf{M}_{\mathbb{A}_H}} \geq \alpha_H OPT(\frac{m}{2},B) $. Combining these two observations along with the inequality in \eqref{eqn:claim}, we obtain the desired approximation ratio of $\frac{1}{2}\min \{\alpha_L, \alpha_H\}$ for $OPT(m, A\cup B)$.\\

Therefore, to complete the proof it remains to show that the inequality in \Cref{eqn:claim} holds. To that end, let the optimal menu to the original problem be given by $\Mset=\menuset$, that is, the expected number of matches achieved by $\Mset=\menuset$ is equal to $OPT(m, \otsupp = A \cup B)$.  

Construct the modified menu set as follows. Let ${\bar{M}}^A$ be such that $M^A_i = M_i \cap A$. Then, for a given customer $i$ and any supplier $j \in A$ such that $j \in A$, we have that under the menu set $\bar{M}$, the probability that $i$ chooses $j$ is given by
\begin{align*}
    p_{i,j} \triangleq \frac{v_j}{1 + \sum_{j ^\prime \in M_i}v_{j^\prime}}.
\end{align*}

\noindent and under the menu set $\bar{M}$ the same probability is
\begin{align*}
    p^A_{i,j} \triangleq \frac{v_j}{1 + \sum_{j ^\prime \in M_i \cap A}v_{j^\prime}}. 
\end{align*}

Then, it is immediate to see that 
\begin{align} \label{eqn:prob_combine}
p^A_{i,j} \geq p_{i,j} \mbox{~for all~} j \in A \mbox{ and all } i \mbox{ such that } j \in M_i. 
\end{align}

Let $p_j(l)$ and $p_j^A(l)$ be the probability that $l$ customers choose supplier $j$ under ${\bar{M}}$ and ${\bar{M}}^A$, respectively. As a consequence of Equation \eqref{eqn:prob_combine}, we must have $p_j(l) \leq p_j^A(l)$ for all possible values of $l$. \\

Let $\Yj$ (resp. $\Yja$) be an indicator variable denoting whether $j \in A$ gets matched or not under ${\bar{M}}$ (resp. ${\bar{M}}^A$). We can then write
\begin{align*}
    \expt{\Yj}  = \expt{\expt{\Yj|\Xj}} = \expt{ \frac{\Xj}{\qj + \Xj}}\quad   \mbox{ and } \quad  \expt{\Yja} = \expt{\expt{\Yja|\Xja}}= \expt{ \frac{\Xja}{\qj + \Xja}},
\end{align*}

\noindent where $\Xj$ (resp.~$\Xj^A$) are random variables representing the number of customers that choose supplier $j$. These random variables can be written as 
\begin{align}
\label{eq:xdefn}
    \Xj = \sum_i \Xij \quad \mbox{ and } \quad 
    \Xja = \sum_i \Xij^A,
\end{align}
with $\Xij$ (resp. $\Xij^A$) representing an indicator random variable denoting whether customer $i$ chooses supplier $j$ or not. 
From  \Cref{eq:xdefn,eqn:prob_combine},  we can see that $\Xj^{A}$ stochastically dominates $\Xj$ for all suppliers $j \in A$. 
Consequently $\expt{\Yj^A} \geq \expt{\Yj}$. 

Then, we have that
\begin{align}
\label{eq:dominanceA}
\sum_{j \in A} \expt{\Yj} \leq \sum_{j \in A} \expt{\Yj^A}.
\end{align}

Analogously, we can repeat the above exercise for set $B$ to ultimately conclude that 
\begin{align}
\label{eq:dominanceB}
\sum_{j \in B} \expt{\Yj} \leq \sum_{j \in B} \expt{\Yj^B}.
\end{align}

 Let $OPT_A$ and $OPT_B$ denote the expected number of matches achieved by the suppliers in $A$ and $B$ respectively under the optimal menu ${\bar{M}}$. Then, 
\begin{align}
\label{eq:decomposition}
OPT(m, A \cup B) = OPT_A + OPT_B.
\end{align}
Combining  \Cref{eq:decomposition,eq:dominanceA,eq:dominanceB} we have,

\begin{align*}
    OPT(m ,A \cup B ) &= OPT_A + OPT_B \\
    &=  \sum_{j \in A} \expt{\Yj} + \sum_{j \in B} \expt{\Yj} \\
    &\leq \sum_{j \in A} \expt{\Yj^A} + \sum_{j \in B} \expt{\Yj^B} \\
    & \leq OPT(m,A) + OPT(m,B) .
\end{align*}

We show next that
\begin{align}\label{eqn:second}
    OPT(m,A) + OPT(m,B) \leq 2 \left[ OPT\left(\frac{m}{2},A\right) + OPT\left(\frac{m}{2},B\right) \right].
\end{align}

To prove the above inequality, it suffices to show that $OPT(t,X) \leq \frac{t}{t-1} OPT(t-1,X)$, for any $t \geq 2$ and any given set of suppliers $X$. 
To see why this is true, let $I$ denote the instance with $t$ customers and a set of supplier given by $X$, where  customers are shown an optimal assortment $Z$ that achieves $OPT(t,X)$.


Fix a customer $k$, discard that customer and consider the resulting number of matches in an instance $I'_k$ where $X$  is the set of suppliers and the $t-1$  remaining customers are shown the menu in $Z$.
Let $p_k(t,X)$ denote the probability that customer $k$ is matched in instance $I$. We will show that $OPT(t,X) \leq OPT(t-1,X) + p_k(t,X)$. As $k$ is arbitrary, the inequality holds for all customers.

With some abuse of notation, we override $\Yj$ and $\Yj'$ to be indicator variables denoting whether $j \in X$ gets matched in instances $I$ and $I'$, respectively. 
For all $i \in \otcust$, let  $\Xij$ for all $j\in X$ and $\Xij'$ for all $j\in X$  be the indicator random variable denoting if $i$ picks $j$ for instance $I$ and $I'$, respectively. Note that $\Xij$ and  $\Xij'$ are identically distributed for all $i \neq k$, as consumers see the same assortments in both instances and consumer choices are independent. 

Further for each $j \in X$ we have that:

$$\expt{\Yj|\Xij \text{ for all } i\in \otcust}=\frac{\sum_{i}\Xij}{\qj+\sum_{i}\Xij},
\text{ and }
\expt{\Yj}=\expt{\frac{\sum_{i}\Xij}{\qj+\sum_{i}\Xij}}.$$

Moreover, if $j \notin Z_k$, we have $\expt{\Yj}=\expt{\frac{\sum_{i}\Xij}{\qj+\sum_{i}\Xij}} = \expt{\frac{\sum_{i}\Xij'}{\qj+\sum_{i}\Xij'}} = \expt{\Yj'}$.

On the other hand, if $j \in Z_k$, we have 
\begin{eqnarray*}
 \expt{\Yj} & = & \expt{\frac{\sum_{i}\Xij}{\qj+\sum_{i}\Xij}} \\
 & \leq & \expt{\frac{\sum_{i\neq k }\Xij}{\qj+\sum_{i \neq k}\Xij  } + \frac{\Xkj}{\qj+\sum_{i}\Xij  } }\\
  & = & \expt{\frac{\sum_{i\neq k }\Xij'}{\qj+\sum_{i \neq k}\Xij'  }} + \expt{ \frac{\Xkj}{\qj+\sum_{i}\Xij  } }\\ 
 & = &  \expt{\Yj'}    + \expt{ \frac{\Xkj}{\qj+\sum_{i}\Xij  } } 
\end{eqnarray*}

Therefore, 

$$ OPT(t,X) = \sum_{j\in X} \expt{\Yj} \leq \sum_{j\in X} \expt{\Yj'} + \sum_{\{j:~j\in Z_k\}} \expt{ \frac{\Xkj}{\qj+\sum_{i}\Xij  } }   \leq OPT(t-1,X) + p_k(t,X).$$

 Therefore, by discarding a customer $k^* = \arg \min_{k} p_k(t,X)$, and maintaining the same assortment that achieves $OPT(t,X)$, the expected number of ensuing matches is at least $\frac{t-1}{t} OPT(t,X)$. Consequently, we have $OPT(t,X) \leq \frac{t}{t-1} OPT(t-1,X)$, and thus we obtain, as desired \eqref{eqn:second}.
\end{proof}

\section{Proofs omitted in Section \ref{subsec:smallSupplierLargeOutside}}
\label{appendix_general}
\renewcommand{\Xij}{\textbf{X}_{ij}}
\renewcommand{\Xijp}{\textbf{X}'_{ij}}
\newcommand{\Xjp}{\textbf{X}'_{j}}
\renewcommand{\Yj}{\textbf{Y}_{j}}
\newcommand{\Yjp}{\textbf{Y}_{j}'}
\newcommand{\Mp}{\textbf{M}'}
\renewcommand{\M}{\textbf{M}}
\newcommand{\polym}{\textbf{C}}

\lemcaseone*
\begin{proof}
 Given the menu set $\menuset$ for customers, let $\Xj$ and $\Xj'$ for all $j \in \otsupp$ be random variables counting the number of customers that match to supplier $j$ in original and new setting, respectively. Similarly let $\Yj$ and $\Yj'$ for all $j \in \otsupp$ be random variables that denotes if supplier $j$ is matched  in the original and new setting respectively. Then, 
 \begin{equation}\label{eq:case2expt}
     \expt{\textbf{M}}=\sum_{j=1}^{\supp}\expt{\Yj}\quad \mbox{and} \quad \expt{\textbf{M}'}=\sum_{j=1}^{\supp}\expt{\Yj'}.
 \end{equation}
 
 Note that the scores $\vj$ and $\vpj$ for all $j \in \otsupp$ are the same in both settings. Therefore, for all $j \in \otsupp$, the random variables $\Xj$ and $\Xj'$ are identically distributed. 
 \renewcommand{\Xjp}{\textbf{X}'_{j}}
 Further for all $j \in \otsupp \backslash S$, the outside options $\qj$ and $\qpj$ are the same and therefore $\expt{\Yj}=\expt{\Yj'}$. 
 
 We consider suppliers $j\in S$ and recall $\qj <1$ and $\qj'=1$ for all these suppliers.
 Note, $\expt{\Yj}=\expt{\frac{\Xj}{\Xj+\qj}}$ and $\expt{\Yj'}=\expt{\frac{\Xj'}{\Xj'+1}}$. Since $\Xj$ (and $\Xj'$) takes only non-negative integral value, we have $\frac{\Xj}{\Xj+1}\leq \frac{\Xj}{\Xj+\qj} \leq 2 \frac{\Xj}{\Xj+1}$ for all $\Xj \in \Z_{\geq 0}$ because $q_j \leq 1$. Combining the above observations together, we get:
 $$\expt{\frac{\Xj}{\Xj+1}} \leq \expt{\Yj}\leq 2\expt{\frac{\Xj}{\Xj+1}}\textbf { and }\expt{\Yj'}=\expt{\frac{\Xj'}{\Xj'+1}}.$$
 
 Since $\Xj$ and $\Xj'$ are identically distributed, we have:
 $$\expt{\Yj'}\leq \expt{\Yj} \leq 2\expt{\Yj'}.$$
 Summing over all $j \in \otsupp$, and using Equation \eqref{eq:case2expt}, we get
 $$\expt{\textbf{M}'} \leq \expt{\textbf{M}}\leq 2\expt{\textbf{M}'}$$
as desired. 
\end{proof}

\begin{lemma}\label{lem:polym}
Let $I$ be an instance of the problem with $\cust$ customers and $\supp$ suppliers, where $\vj \in {\R}$ and $\qj \in {\R}$ denote the score and outside option for supplier $j \in \otsupp$. Let $I'$ be the instance obtained from $I$ by removing suppliers $j$  with $\qj \geq \polym$. If $\opt$ and $\opt'$ denote the expected size of the matching under the  optimal menus for instances $I$ and $I'$ respectively, then:
$$\opt'\geq \opt-\frac{\cust}{\polym}.$$
\end{lemma}
\begin{proof}
Let $S$ and $S'$ be the set of all suppliers in instance $I$ and $I'$ respectively. Let $\mathrm{M}_i$ be the menu for the $i$'th customer in the optimal menu for instance $I$. Define $\mathrm{M}'_i\defeq \mathrm{M}_i \cap S'$ be the menu for $i$'th customer for instance $I'$. Let $\Yj$ for $j \in S$ and $\Yj^\prime$ for $j\in S'$ be random variables denoting if supplier $j$ is matched under instance $I$ and $I'$ respectively. Let $\M$ and $\Mp$ be the random variable denoting the number of suppliers matched under instances $I$ and $I'$ respectively. Then $\M=\sum_{j\in S}\Yj$ and $\Mp=\sum_{j\in S'}\Yj^\prime$. For all $i \in \otcust$, let $\Xij$ for all $j\in S$ and $\Xijp$ for all $j\in S'$  be the indicator random variable denoting if $i$ picks $j$ for instance $I$ and $I'$ respectively. Define $\alpha_i\defeq \sum_{j\in \M_i}\vj$ and $\alpha'_i\defeq \sum_{j\in \Mp_i}\vj$. Note for all $j\in S'$, $\prob{\Xij=1}=\frac{\vj}{1+\massi}$ and $\prob{\Xijp=1}=\frac{\vj}{1+\alpha'_i}$. Since $\alpha_i \geq \alpha'_i$ and $\Xij,\Xijp$ are indicator random variables, we have that $\Xijp$ stochastically dominates $\Xij$. Further since all $\Xij$ and $\Xijp$ for all $i\in \otcust$ are independent random variables, we have that $\Xjp=\sum_{i\in \otcust}\Xijp$ stochastically dominates  $\Xj=\sum_{i\in \otcust}\Xij$ (using \Cref{lem:stocdom} inductively). This further implies $\expt{\Yj^\prime} \geq \expt{\Yj}$, because $\Yj$ and $\Yjp$ are monotone functions in $\Xj$ and $\Xjp$ respectively. Further for each $j \in S \backslash S'$,
$$\expt{\Yj|\Xij \text{ for all } i\in \otcust}=\frac{\sum_{i}\Xij}{\qj+\sum_{i}\Xij},
\text{ and }
\expt{\Yj}=\expt{\frac{\sum_{i}\Xij}{\qj+\sum_{i}\Xij}}.$$
Using the above expression, we now provide an upper bound for the quantity $\sum_{j\in S\backslash S'}\expt{\Yj}$,
\begin{equation}\label{eq:upp}
    \begin{split}
        \sum_{j\in S\backslash S'}\expt{\Yj}=\sum_{j\in S\backslash S'}\expt{\frac{\sum_{i}\Xij}{\qj+\sum_{i}\Xij}}\leq \sum_{j\in S\backslash S'} \expt{\frac{\sum_{i}\Xij}{\polym}} \leq \frac{\cust}{\polym}~.
    \end{split}
\end{equation}

For the final inequality, note that $\sum_{j \in S\backslash S'}\sum_{i}\Xij$ is always upper bounded by the number of customers $\cust$. Now we are ready to lower bound the value of $\opt'$,
$$\opt' =\sum_{j \in S'}\expt{\Yj^\prime} \geq \sum_{j \in S}\expt{\Yj}-\sum_{j \in S\backslash S'}\expt{\Yj} \geq \opt -\frac{\cust}{\polym}.$$

For the second inequality, we used the fact that $\expt{\Yjp} \geq \expt{\Yj}$ for all $j \in S'$. For the last inequality, we used Equation \eqref{eq:upp}.
\end{proof}
\subsection{Proofs omitted in \Cref{sec_lp_relax_gen}}
\label{appendix:lp_relaxation_general}
For convenience, we start by restating optimization problem \eqref{eq:foptg} 
    \begin{equation*}
        \begin{split}
     \fopt= & \max_{x} \sum_{k \in \otbuck}\min \left \{\frac{2}{\qr}\sum_{i \in \otcust}\frac{\wl \xik}{1+\sum_{k' \in \otbuck}\wlp \xikp},|\Sk| \right \}\\
    \text{s.t.}& ~ 0 \leq \xik \leq |\Sk|\text{ for all }i\in \otcust \text{ and }k\in \otbuck ~.
        \end{split}
    \end{equation*}
     
\lemfoptg*
\begin{proof}
    Let $\{\Mstari\}_{i \in \otcust}$ be the optimal menu, where  $\Mstari$ is menu for customer $i$. For each bucket $k \in \otbuck$,  define $\xstar_{i,k}$ to be the number of suppliers from bucket $k$ present in menu $\Mstari$. Here we aim to upper bound the expected number of matches achieved by optimal menu. 
    
    Let $\Xk$ be the random variable counting the number of customers that choose a supplier in bucket $k$, given $\{\Mstari\}_{i \in \otcust}$.  Then, by \Cref{eq:gexptxk}, we have
       
        $$\expt{\Xk}\leq 2 \sum_{i=1}^{\cust}\frac{\wl \xstar_{i,k}}{1+\sum_{k' \in \otbuck}\wlp \xstar_{i,k'}}~.    $$
    Let $\Yk$ be the random variable counting the number of matches in bucket $k$ for the optimal menu. Then the expected number of matches from bucket $k$ is equal to $\expt{\Yk}$ and, by \Cref{eq:exptykub}, this quantity is upper bounded by:
    \begin{equation*}
        \expt{\Yk} \leq \min \left \{\frac{2}{\qr}\sum_{i \in \otcust}\frac{\wl \xstar_{i,k}}{1+\sum_{k' \in \otbuck}\wlp \xstar_{i,k'}},|\Sk| \right \}. 
    \end{equation*}
    
    Further note that $0 \leq \xstar_{i,k} \leq |\Sk|$ and is a feasible solution to optimization problem \eqref{eq:foptg}. The lemma statement then follows by summing over all buckets $k$.
    \end{proof}
    
    In the remaining part of this section we provide the proof for \Cref{lem:lpoptg} by constructing a sequence of relaxations to problem \eqref{eq:foptg}. To that end, consider the following optimization problem.
    \begin{equation}\label{eq:optoneg}
    \begin{split}
     \optoneg  = &  \max_{x} \sum_{k} \min \left \{\frac{2}{\qr}\sum_{i \in \otcust}\frac{\wl \xik}{1+\sum_{k'}\wlp \xikp},|\Sk| \right \} \\
        \text{s.t.}& ~ \sum_{k' \in \otbuck}\wlp \xikp \leq 1 \text{ for all }i\in \otcust,\\ 
     &~0 \leq \xik \leq |\Sk|\text{ for all }i\in \otcust \text{ and }k\in \otbuck ~.
    \end{split}
    \end{equation}

    \begin{lemma}
    \label{lem:optoneg}  $\optoneg \leq \fopt \leq 2 \optoneg $.
    \end{lemma}
    \begin{proof}
    The first inequality $\optoneg \leq \fopt$ holds because the feasible region in problem \eqref{eq:optoneg} is contained in the one of problem \eqref{eq:foptg}, and both problems have the same objective.
    
    To establish the second inequality, let $\xstar$ be the optimal solution to optimization problem \eqref{eq:foptg} and let $\xstar_{i}\defeq \sum_{k'}\wlp \xstar_{i,k'}$. We now construct a feasible solution $y$ for  problem \eqref{eq:optoneg} as follows: For each $i \in \otcust$,

    \begin{equation*}
        \text{If }\xstar_{i} \leq 1, \text{ then define  }y_{i,k}\defeq\xstar_{i,k} \text{ for all }k \in \otbuck.
    \end{equation*}
    \begin{equation*}
        \text{If }\xstar_{i} >1, \text{ then define  }y_{i,k}\defeq\xstar_{i,k}\frac{1}{\xstar_{i}} \text{ for all }k \in \otbuck.
    \end{equation*}
Let $\yi\defeq\sum_{k\in \otbuck}\yik$. Now for each $i \in \otcust$, $\yi \leq 1$ as

    \begin{equation*}
        \text{If }\xstar_{i} \leq 1, \text{ then }\yi=\sum_{k\in \otbuck}\wk y_{i,k}=\sum_{k\in \otbuck}\wk\xstar_{i,k}=\xstar_{i}\leq 1.
    \end{equation*}

    \begin{equation*}
        \text{If }\xstar_{i} >1, \text{ then }\yi=\sum_{k\in \otbuck}\wk y_{i,k}=\sum_{k\in \otbuck}\left(\wk \xstar_{i,k}\frac{1}{\xstar_{i}}\right)=1.
    \end{equation*}
    Therefore, $y$ is a feasible solution for  problem \eqref{eq:optoneg} and, further, for each $i \in \otcust$ consider each term in the objective, 
    \begin{equation*}
        \text{If }\xstar_{i} \leq 1, \text{ then }\frac{\wk \yik}{1+\yi}=\frac{\wk \xstar_{i,k}}{1+\xstar_{i}}.
    \end{equation*}
    \begin{equation*}
        \text{If }\xstar_{i} > 1, \text{ then }\frac{\wk \yik}{1+\yi}=\frac{\wk \frac{\xstar_{i,k}}{\xstar_{i}}}{1+1}=\frac{\wk\xstar_{i,k}}{2\xstar_{i}} \geq \frac{\wk\xstar_{i,k}}{2(1+\xstar_{i})}. 
    \end{equation*}
    The lemma statement follows because $y$ approximates $\xstar$ by $1/2$ for each summation term in the objective.
        \end{proof}
        
         Next we provide the second optimization problem that approximates optimization problem \eqref{eq:optoneg} and in the following lemma we provide the explicit approximation factor. 
         
    \begin{equation}\label{eq:opttwog}
        \begin{split}
     \opttwog= & \max_{x}             \sum_{k\in \otbuck}\min\left \{ \frac{2}{\qr} \sum_{i=1}^{\cust}  \wk \xik, |\Sk|\right \} \\
        \text{s.t.}& ~ \sum_{k' \in \otbuck}\wlp \xikp \leq 1 \text{ for all }i\in \otcust,\\ 
     &~0 \leq \xik \leq |\Sk|\text{ for all }i\in \otcust \text{ and }k\in \otbuck ~.
        \end{split}
    \end{equation}
    
    \begin{lemma}[{$\optoneg$ and $\opttwog$ approximate each other}]\label{lem:opttwog}  $\optoneg\leq \opttwog\leq 2  \optoneg$.
    \end{lemma}
    \begin{proof}The left hand side inequality $\optoneg\leq \opttwog$ is trivial and follows by comparing the objectives in both problems. We now show the remaining inequality. Under the constraint $\sum_{k'\in \otbuck}\wkp \xikp \leq 1$ for all $i\in \otcust$, we have 
    $$ \frac{1}{2\qr} \sum_{i=1}^{\cust}\wk\xik \leq \frac{1}{\qr} \sum_{i=1}^{\cust}\frac{\wk\xik}{1+\sum_{k'\in \otbuck}\wkp \xikp} \text{ for all } k \in \otbuck.$$ 
    
    The above inequality further implies that,  for each $k \in \otbuck$ 
    $$\min\left \{ \frac{2}{\qr} \sum_{i=1}^{\cust}\frac{\wk\xik}{1+\sum_{k'\in \otbuck}\wkp \xikp}, |\Sk|\right \} \geq \min\left \{ \frac{2}{2\qr}\sum_{i=1}^{\cust} \wk\xik, |\Sk|\right \}
    \geq \frac{1}{2}\min\left \{\frac{2}{\qr}\sum_{i=1}^{\cust} \wk\xik, |\Sk|\right \} $$
  The proof follows by summing the above equation for all $k \in \otbuck$.
    \end{proof}
    We now recall our LP formulation that upper bounds the expected number of matches  in the optimal menu. Further, in our following lemma which show that optimum values of LP formulation and optimization problem \eqref{eq:opttwog} are equal. Now recall optimization problem \eqref{eq:lpoptg},
       \begin{equation*}
        \begin{split}
     \lpoptg = &   \max_{x} \sum_{k} \frac{2}{\qr}\sum_{i \in \otcust}\wl \xik \\
   \text{s.t.}& ~ \sum_{k' \in \otbuck}\wlp \xikp \leq 1 \text{ for all }i\in \otcust,\\ 
     &~\frac{2}{\qr}\sum_{i=1}^{m}\wl \xik \leq |\Sk| \text{ for all }k \in \otbuck \\
     &~0 \leq \xik \leq |\Sk|\text{ for all }i\in \otcust \text{ and }k\in \otbuck ~.
        \end{split}
    \end{equation*}
   
    \begin{lemma}\label{lem:twolpg}
     $\opttwog = \lpoptg$
    \end{lemma}
    
    \begin{proof}Clearly, $\lpoptg \leq \opttwog$. Next we show the other inequality. Let $\xstar$ be the optimum solution for $\opttwog$. Now consider the following solution,
    $$y_{i,k}=\xstar_{i,k} \frac{\qr}{2}\frac{|\Sk|}{\sum_{i'=1}^{\cust}\wk\xstar_{i',k} } \text{ for all }i \in \otcust \text { if } \frac{2}{\qr} \sum_{i'=1}^{\cust}\wk\xstar_{i',k} > |\Sk|$$
    $$y_{i,k}=\xstar_{i,k} \text{ for all }i \in \otcust \text { if } \frac{2}{\qr} \sum_{i'=1}^{\cust}\wk\xstar_{i',k} \leq |\Sk|.$$
    Since $|\Sk|\geq 1$, it is not hard to see that $y_{i,k} \leq \xstar_{i,k}$ for all $i \in \otcust$ and $k\in \otbuck$. Further for each bucket $k\in \otbuck$, 
    \begin{equation}\label{eq:1}
    \text{If }\frac{2}{\qr} \sum_{i=1}^{\cust}\wk\xstar_{i,k} \leq |\Sk|, \text{ then } \frac{2}{\qr} \sum_{i=1}^{\cust}\wk y_{i,k}= \frac{2}{\qr}\sum_{i=1}^{\cust}\wk\xstar_{i,k} \leq |\Sk|
    \end{equation}
    
    \begin{equation}\label{eq:2}
    \text{If } \frac{2}{\qr} \sum_{i=1}^{\cust}\wk\xstar_{i,k} > |\Sk|, \text{ then } \frac{2}{\qr} \sum_{i=1}^{\cust}\wk y_{i,k}=\sum_{i=1}^{\cust} \wk \xstar_{i,k}\frac{|\Sk|}{\sum_{i'=1}^{\cust}\wk\xstar_{i',k}}=|\Sk|
    \end{equation}
    
   \noindent Therefore $y$ is a feasible solution to  $\lpoptg$ and its objective value is,
    \begin{equation}
    \begin{split}
    \sum_{k\in \otbuck} \frac{2}{\qr}\sum_{i=1}^{\cust}\wk y_{i,k}&=\sum_{\{k~|~ \frac{2}{\qr}\sum_{i=1}^{\cust}\wk\xstar_{i,k} > |\Sk| \}}\frac{2}{\qr}\sum_{i=1}^{\cust}\wk y_{i,k}+\sum_{\{k~|~ \frac{2}{\qr} \sum_{i=1}^{\cust}\wk\xstar_{i,k} \leq |\Sk| \}}\frac{2}{\qr}\sum_{i=1}^{\cust}\wk y_{i,k}\\
    &=\sum_{\{k~|~ \frac{2}{\qr}\sum_{i=1}^{\cust}\wk\xstar_{i,k} > |\Sk| \}} |\Sk|+\sum_{\{k~|~\frac{2}{\qr} \sum_{i=1}^{\cust}\wk\xstar_{i,k} \leq |\Sk| \}} \frac{2}{\qr}\sum_{i=1}^{\cust}\wk \xstar_{i,k}\\
    &=\sum_{k\in \otbuck}\min \left \{ \frac{2}{\qr}\sum_{i=1}^{\cust}\wk \xstar_{i,k} ,|\Sk| \right \}
    \end{split}
    \end{equation}
    The second equality follows from equations  \eqref{eq:1} and \eqref{eq:2}. Therefore $\lpoptg \geq \opttwog$ as desired. 
    \end{proof}
    
    \finaloptg*
    
    \begin{proof}
    The lemma follows by combining  \cref{lem:optoneg}, \cref{lem:opttwog} and \cref{lem:twolpg}.
    \end{proof}
    

\subsection{Proofs omitted in \Cref{rounding_general}}
\label{appendix:rounding_menu_general}
We first state an observation that will be useful for the proofs in this section.
\begin{observation}[\Cref{lem:polym}]\label{obs:remove} 
We can assume $|\Bt|=\cust$ by ignoring all suppliers with $\qj \geq 2^{\cust}$ and only incur a loss of additive $\frac{\cust}{2^{\cust}}$ in expected size of the matching. Thereby, we have that
  \begin{align*}
      \sum_{\ell \in \Bo}\ell &\leq 2 \quad \text{ and } \quad \sum_{k\in \otbuck}\wl \leq |\Bt|\sum_{\ell \in \Bo}\ell \leq 2\cust.
  \end{align*}
\end{observation}

\rounding*
 \begin{proof}
 \renewcommand{\cm}{c}
Let $\xstar$ be an optimal solution to \eqref{eq:lpoptg}, and $\lpoptg$ be the corresponding optimal value, and let $x$ be the corresponding output of the rounding algorithm. For each bucket $k$,
    \begin{equation}\label{eq:k1}
        \begin{split}
            \sum_{i=1}^{\cust} \wl \xik&=\sum_{\{i:~\xstar_{ik}<1 \}}\wl \xik+\sum_{\{i:~\xstar_{ik}\geq 1 \}}\wl \xik, 
            \geq \wl \ceil{\sk}+\frac{1}{2}\sum_{\{i:~\xstar_{ik}\geq 1 \}}\wl \xstar_{ik} \\
            &\geq \wl \sum_{\{i:~\xstar_{ik}<1 \}} \xstar_{ik} + \frac{1}{2}\sum_{\{i:~\xstar_{ik}\geq 1 \}}\wl \xstar_{ik}
            \geq\frac{1}{2} \sum_{i=1}^{\cust}\wl\xstar_{ik}~. 
        \end{split}
    \end{equation}
    The second inequality holds because $\xik \geq \frac{1}{2}\xstar_{i,k}$ for all $i,k$ such that $\xstar_{ik}\geq 1$, and $\sum_{\{i:~\xstar_{ik}<1 \}}\xik=\ceil{\sk}$. For the third inequality, note that $\sk=\sum_{\{i:~\xstar_{ik}<1 \}}\xstar_{i,k}$ and $\sk \leq \ceil{\sk}$.
    
    Now, multiplying both sides of Equation \eqref{eq:k1} by $2/\qr$, we get $\frac{2}{\qr}\sum_{i} \wl \xik \geq \frac{1}{2}\frac{2}{\qr} \sum_{i} \wl \xstar_{i,k}$ that further implies $\min \{\frac{2}{\qr} \sum_{i} \wl \xik,|\Sk| \} \geq \frac{1}{2} \min \{\frac{2}{\qr}\sum_{i} \wl \xstar_{i,k},|\Sk|\}$. Now summing over all buckets $k \in \otbuck$, we get:
    $$\sum_{k \in \otbuck}\min \left\lbrace \frac{2}{\qr}\sum_{i=1}^{\cust} \wl \xik,|\Sk| \right\rbrace \geq \frac{1}{2}\sum_{k \in \otbuck}\frac{2}{\qr}\sum_{i=1}^{\cust} \wl \xstar_{i,k} = \lpoptg/2~.$$
    The first inequality follows because $\frac{2}{\qr}\sum_{i} \wl \xstar_{i,k} \leq |\Sk|$ is a constraint in problem \eqref{eq:lpoptg}.
    
    We next provide a proof for the second condition. Given an assignment $x$, we denote $\tm=\sum_{i=1}^{\cust}\sum_{\{k \in \otbuck~|~\xstar_{i,k} <1 \}}\wl\xik$ and for each $l\in \Bo$ define $\ml\defeq \sum_{i=1}^{\cust}\sum_{\{k\in \otbuck\in \otbuck~|~\wl=\ell \text{ and }\xstar_{i,k}<1\}}\xik$ and note that $\ml=\sum_{\{k\in \otbuck \text{ and }\wl=\ell  \}}\ceil{\sk}$. Furthermore, $\tm=\sum_{l \in \Bo}l*\ml$. We now provide an upper bound for $\tm$: 
    \begin{equation}\label{eq:tm}
        \begin{split}
    \tm=\sum_{i=1}^{\cust}\sum_{\{k \in \otbuck~|~\xstar_{i,k} <1 \}}\wl\xik&=\sum_{k \in \otbuck}\wl\sum_{\{i \in \otcust~|~\xstar_{i,k} <1 \}}\xik=\sum_{k \in \otbuck }\wl\ceil{\sk}\\
    & \leq \sum_{k \in \otbuck}\wl\sum_{\{i \in \otcust~|~\xstar_{i,k} <1 \}}\xstar_{i,k}+\sum_{k\in \otbuck}\wl \leq \cust+2\cust= 3\cust.
        \end{split}
    \end{equation}
    In the final expression, we exchanged the orders of summation. Further we used Observation \ref{obs:remove} and also used the fact that for each customer $i \in \otcust$, $\sum_{k \in \otbuck}\wl \xstar_{i,k} \leq 1$, which is a constraint in problem \ref{eq:lpoptg}. 
    
    Now fix $l \in \Bo$, for each $i \in \otcust$, we next give an upper bound for the quantity $\sum_{\{k \in \otbuck~|~\wl=\ell \text{ and }\xstar_{i,k}<1 \}}\wl \xik$.
    \begin{equation}
        \begin{split}
            \sum_{\{k \in \otbuck~|~\wl=\ell \text{ and }\xstar_{i,k}<1 \}}\wl \xik=\ell*\sum_{\{k \in \otbuck~|~\wl=\ell \text{ and }\xstar_{i,k}<1 \}}\xik\leq \ell*\left(\frac{\ml}{m}+1\right)~.
        \end{split}
    \end{equation}
    The final inequality in the above expression follows from the assignment of the \emph{rounding} algorithm. We are now ready to prove the second condition. For each customer $i \in \otcust$,
    \begin{equation}
    \begin{split}
        \sum_{k \in \otbuck}\wl \xik &=\sum_{\{k\in \otbuck~|~\xstar_{i,k}<1 \}}\wl \xik + \sum_{\{k\in \otbuck~|~\xstar_{i,k} \geq 1 \}}\wl \xik,\\
        &=\sum_{\ell \in \Bo}\ell\sum_{\{k\in \otbuck~|~\wl=\ell \text{ and }\xstar_{i,k}<1 \}} \xik + 1
        =\sum_{\ell \in \Bo}\ell \yil + 1, \\
        &\leq \sum_{\ell \in \Bo}\ell*\left(\frac{\ml}{m}+1\right)+1
        =\frac{\tm}{m}+\sum_{\ell \in \Bo}\ell+1\leq \cm + o(1).
    \end{split}
    \end{equation}
    In the second equality, we rearranged the summation and used the fact that $\sum_{\{k\in \otbuck:~\xstar_{i,k} \geq 1 \}}\wl \xik = \sum_{\{k:~\xstar_{i,k} \geq 1 \}}\wl \floor{\xstar_{i,k}} \leq\sum_{\{k:~\xstar_{i,k} \geq 1 \}}\wl \xstar_{i,k} \leq 1$ (constraint of optimization problem \ref{eq:lpoptg}). In the third equality, we used the fact that $\sum_{\{k\in \otbuck:~\wl=\ell \text{ and }\xstar_{i,k}<1 \}} \xik=\yil$ for all $\ell \in \Bo$ and $i \in \otcust$ and this follows from the assignment of \emph{rounding} algorithm. For the fourth inequality, note that $\yil\leq 1+ \frac{1}{\cust}\sum_{\{k\in \otbuck \text{ and }\wl=\ell  \}}\ceil{\sk}=1+\frac{1}{\cust}\ml$. In the final inequality, we used $\frac{\tm}{\cust}+\sum_{\ell \in \Bo}\ell+1\leq 1+O(\frac{\log \cust}{\cust})+2+1\leq \cm$ for some constant $\cm$.

    All that remains to show is the third condition: $\frac{2}{\qr}\sum_{i \in \otcust}\wl \xik \leq |\Sk|+\frac{2\wl}{\qr}$.
    \begin{equation}
        \begin{split}
            \frac{2}{\qr}\sum_{i \in \otcust}\wl \xik & = \frac{2\wl}{\qr}\sum_{\{i:~\xstar_{ik}<1 \}} \xik+\frac{2\wl}{\qr}\sum_{\{i:~\xstar_{ik}>1 \}} \xik
            \leq  \frac{2\wl}{\qr}\ceil{\sk}+\frac{2\wl}{\qr}\sum_{\{i:~\xstar_{ik}\geq 1 \}} \xstar_{ik} ,\\
            &\leq \frac{2\wl}{\qr}(\sk+1)+\frac{2\wl}{\qr}\sum_{\{i:~\xstar_{ik}\geq 1 \}} \xstar_{ik} 
            \leq \frac{2\wl}{\qr} + |\Sk|.
        \end{split}
    \end{equation}
    The first inequality holds because $\xik \leq \xstar_{i,k}$ for all $i,k$ such that $\xstar_{ik}\geq 1$, and $\sum_{\{i:~\xstar_{ik}<1 \}}\xik=\ceil{\sk}$. For the second inequality, note that $\ceil{\sk}\leq \sk+1$. In the final inequality, we used the fact that $\sk=\sum_{\{i:~\xstar_{ik}<1 \}} \xstar_{i,k}$, combined with $ \frac{2\wl}{\qr}\sum_{\{i:~\xstar_{ik}< 1 \}} \xstar_{ik} +\frac{2\wl}{\qr}\sum_{\{i:~\xstar_{ik}\geq 1 \}} \xstar_{ik} \leq |\Sk|$, which is a constraint in optimization problem \eqref{eq:lpoptg}.
    \end{proof}

\menuconstruction*
\begin{proof}
The first condition follows immediately from the construction of menus. For the second, observe that $\sum_{i}{\xik} \leq \frac{\qr}{2\wl}|\Sk|+1$, and $\ckj \leq \frac{\sum_{i}\xik}{\Sk}+1 \leq \frac{\qr}{2\wl}+\frac{1}{\Sk}+1\leq 2+\frac{\qr}{2\wl}$.
\end{proof}

\section{Proofs omitted in \Cref{subsec:largeSupplier}}\label{app:casetwo}
\lemcasetwoup*
\begin{proof}
\newcommand{\Ystarj}{\textbf{Y}^{*}_{j}}
\newcommand{\Xstarj}{\textbf{X}^{*}_{j}}
Given optimal menu, let $\Xstarj$ for all $j \in \otsupp$ be the random variable counting the number of customers that select supplier $j$. Further, let $\Ystarj$ be a random variable denoting if supplier $j$ is matched. Then,
\begin{equation}\label{eq:case11}
\opt=\sum_{j=1}^{\supp}\expt{\Ystarj}.
\end{equation}

The expected value of $\expt{\Ystarj}$ can be written as
\begin{equation}\label{eq:case12}
\expt{\Ystarj}=\expt{\expt{\Ystarj|\Xstarj}}=\expt{\frac{\Xstarj}{\Xstarj+\qj}}.
\end{equation}

Combining Equations \eqref{eq:case11} and \eqref{eq:case12} we obtain
\begin{equation}
\begin{split}
    \opt&=\expt{\sum_{j=1}^{\supp}\frac{\Xstarj}{\Xstarj+\qj}} \\
    & \leq \sum_{\{x \in \Z_{\geq 0}^{\supp}:~\sum_{j=1}^{\supp}\xj\leq \cust\}}\prob{\Xstarj=\xj, \text{ for all }j \in \otsupp}\left(\sum_{j=1}^{\supp}\frac{\xj}{\xj+\qj}\right)\\
    &\leq \max_{\{x \in \Z_{\geq 0}^{\supp}:~\sum_{j=1}^{\supp}\xj=\cust\}}\sum_{j=1}^{\supp}\frac{\xj}{\xj+\qj}~.
\end{split}
\end{equation}
The first inequality holds because $\Xstarj$ is a non-negative integer valued random variable, and $\sum_{j=1}^{\supp}\Xstarj\leq m$ with probability 1. The last inequality follows because the function $\sum_{j=1}^{\supp}\frac{\xj}{\xj+\qj}$ is monotone in $\sum_{j=1}^{\supp}\xj$.
\end{proof}

\lemeff*
\begin{proof}
Without loss of generality we can assume $q_1\leq q_2\leq \dots q_n$. Let $\xstar$ be the optimum solution to optimization problem \eqref{eq:case1ub}.
Recall we wish to solve the optimization problem \eqref{eq:case1ub} approximately and it will be convenient to identify indices where $\xstar_{i}=0$. 

It is easy to see that 
if $\xstar_j=0$, then $\xstar_{j'}= 0$ for all $j' \geq j$; otherwise, one could simply swap the values of $\xstar_j$ and $\xstar_{j'}$ and increase the objective. 
Therefore, we state the following optimization problem, indexed by $i$.
\begin{equation}\label{eq:case2ub}
    \opti\defeq \max_{\{x \in \Z_{\geq 0}^{\supp}:~\sum_{j=1}^{i}\xj=\cust \text{ and }x_{j}\geq 1\}}\sum_{j=1}^{i}\frac{\xj}{\xj+\qj}.
\end{equation}
Note the addition of constraint $x_j\geq 1$ in the above optimization problem. Using the definition of $\opti$, the optimization problem \eqref{eq:case1ub} can be equivalently written as:
\begin{equation}
    \max_{i\in \otsupp}\opti
\end{equation}
To solve the optimization problem \eqref{eq:case1ub} approximately (up to a factor of $\frac{1}{2}$) it is enough to solve optimization problem \eqref{eq:case2ub} approximately (up to a factor of $\frac{1}{2}$) for all $i \in \otsupp$. Now define the relaxation of \eqref{eq:case2ub} as follows:
\begin{equation}\label{eq:case3ub}
    \copti\defeq \max_{\{x \in \R^{\supp}:~\sum_{j=1}^{i}\xj=\cust \text{ and }x_{j}\geq 1\}}\sum_{j=1}^{i}\frac{\xj}{\xj+\qj}.
\end{equation}
Note all the constraints are linear and the objective is concave in $x$. Maximizing a concave function over convex set yields a convex optimization problem and  thus can be solved efficiently. Let $x'$ be the optimum solution for convex optimization problem \eqref{eq:case3ub}. Consider the following integral solution:
\begin{equation}
x_{j}\defeq 
    \begin{cases}
    & \floor{x'_{j}} \text{ for all } 2 \leq j \leq m \\
    & \cust-\sum_{j=2}^{i}\floor{x'_{j}} \text{ for }j=1.
    \end{cases}
\end{equation}
Clearly $\xj\geq 1$ and $\sum_{j=1}^{i}\xj=m$, therefore is a feasible solution to optimization problem \eqref{eq:case2ub}. Further note that $x_{1}=\cust-\sum_{j=2}^{i}\floor{x'_{j}} \geq \cust-\sum_{j=2}^{i}x'_j =x'_1$, where in the last inequality we used the constraint $\sum_{j=1}^{i}x'_j=\cust$. Since $x'_j \geq 1$ for all $1 \leq j\leq i$, we have $\xj \geq \frac{1}{2}x'_{j}$, which further implies:
$$\sum_{j=1}^{i}\frac{\xj}{\xj+\qj} \geq \frac{1}{2} \sum_{j=1}^{i}\frac{x'_j}{x'_j+\qj} \geq \frac{1}{2}\opti~.$$
The solution $x$ is a $\frac{1}{2}$-approximate solution to optimization problem \eqref{eq:case2ub} and therefore by iterating over all $i\in \otsupp$ and returning the maximum of all these approximate solutions is a $\frac{1}{2}$-approximate solution to optimization problem \eqref{eq:case1ub}.
\end{proof}

\end{document}